\documentclass[journal,10pt]{IEEEtran}

\usepackage{indentfirst}
\usepackage{booktabs}
\usepackage{graphicx}
\usepackage{subfigure}
\usepackage{booktabs}
\usepackage{array}
\usepackage{cite}
\usepackage{makecell}
\usepackage{algorithm}
\usepackage{amsthm}

\usepackage{algorithmic}

\usepackage{caption}
\usepackage{epsfig}
\usepackage{amsmath}
\usepackage{color}
\usepackage{url}

\usepackage{multirow}
\usepackage{amsfonts}
\usepackage{caption}

\newcommand{\bm}[1]{\mbox{\boldmath{$#1$}}}

\ifCLASSINFOpdf
\else

\fi

\begin{document}

\title{Clustering Based Hybrid Precoding Design for Multi-User Massive MIMO Systems}

\author{\IEEEauthorblockN{Ling~Zhang, Lin~Gui, \emph{Member,~IEEE}, Kai~Ying, Qibo~Qin}



}



\maketitle

\begin{abstract}
Hybrid precoding has been recognized as a promising technology to combat the path loss of millimeter wave signals in massive multiple-input multiple-output (MIMO) systems. However, due to the joint optimization of the digital and analog precoding matrices as well as extra constraints for the analog part, the hybrid precoding design is still a tough issue in current research. In this paper, we adopt the thought of clustering in unsupervised learning and provide design schemes for fully-connected hybrid precoding (FHP) and adaptively-connected hybrid precoding (AHP) in multi-user massive MIMO systems. For FHP, we propose the hierarchical-agglomerative-clustering-based (HAC-based) scheme to explore the relevance among RF chains in optimal hybrid procoding design. The similar RF chains are merged into an individual RF chain when insufficient RF chains are equipped at the base station (BS). For AHP, we propose the modified-K-means-based (MKM-based) scheme to explore the relevance among antennas at the BS. The similar antennas are supported by the same RF chain to make full use of the flexible connection in AHP. Particularly, in proposed MKM-based AHP design, the clustering centers are updated by alternating-optimum-based (AO-based) scheme with a special initialization method, which is capable to individually provide feasible sub-connected hybrid precoding (SHP) design. Simulation results highlight the superior spectrum efficiency of proposed HAC-based FHP scheme, and the high power efficiency of proposed MKM-based AHP scheme. Moreover, all the proposed schemes are clarified to effectively handle the inter-user interference and outperform the existing work.

\end{abstract}


\IEEEpeerreviewmaketitle

\section{Introduction}

\IEEEPARstart{W}{ith} the popularity of intelligent terminals, mobile data traffic is facing exponential growth. To meet the potential capacity requirements for future wireless communication, various novel wireless techniques such as massive multiple-input-multiple-output (MIMO), advanced channel coding, and non-orthogonal multiple access have enthused much attention \cite{MIMO,POLAR,NOMA}. Nevertheless, the bandwidth shortage in physical layer leads to the fundamental bottleneck for capacity improvement \cite{Bandshort}. Thus, it is imminent to develop spectrum bands which have not been utilized in current cellular systems.

Millimeter wave (mmWave) band spanning from 30 to 300 GHz has been determined as the alternative band to expand the available bandwidth in 5G systems \cite{3GPPFR1}. Benefiting from the short wavelength of mmWave, it is feasible to deploy large-scale antennas in limited space at transceivers to implement massive MIMO systems. However, due to the extremely high carrier frequency, mmWave signals experience more serious propagation path loss compared with signals in 3G or LTE. It is necessary to use precoding technology to achieve highly directional beamforming \cite{Pathloss,Necessary}.

In traditional MIMO systems, full digital precoding is the typical scheme to adjust the amplitudes and  phases of the transmit signals \cite{FD1,FD2}. For point-to-point systems, the optimal full digital precoder is directly determined by the singular value decomposition (SVD) of the channel matrix \cite{Sparse}. As for multi-user systems, there are three efficient precoding schemes, including matched-filter (MF), zero-forcing (ZF), and regularized zero-forcing (RZF) methods, to manage the interference among users \cite{MF,ZFR}. However, in the full digital structure, the number of radio frequency (RF) chains is equivalent to that of antennas, which results in prohibitive hardware and power consumption for massive MIMO systems. To tackle this problem, the analog-only procoding scheme is proposed in \cite{AP1,AP2,AP3}. In analog structure, only analog phase shifters (APSs) are utilized to control the phases of the transmit signals, which costs much less than the full digital structure and has been adopted in commercial indoor mmWave communication standards like IEEE 802.11ad \cite{stad}. However, the APSs impose constant modulus constraint on the entries of the precoding matrix, which leads to a less degree of signal freedom and poorer precoding performance compared the full digital precoding \cite{BADAP}.

As a promising precoding scheme for massive MIMO systems, the hybrid precoding architecture has been widely investigated to provide a tradeoff between consumption and performance \cite{HPSERVY}. The hybrid precoding architecture combines the digital precoder in the baseband and the analog precoder in the RF domain. Benefited from the low-dimensional digital precoder, fewer RF chains are required for implementation.

\subsection{Related Works}
Recent research on hybrid precoding focuses on the fully-connected \cite{Sparse,ARVcs,DFT1,DFT2,optfhp,fuldecp1,fuldecp2,fuladd,fuladd2} and the sub-connected structures \cite{subpso,subsnr,subcodbk1,subcodbk2,subcodbk3,subdecp,subml1,subml2}, which can be distinguished by the connection state between RF chains and antennas as illustrated in Fig. 1(a) and Fig. 1(b).

In the fully-connected structure, each antenna is supported by all RF chains through APSs and RF adders. Considering the single-user scenario,
the precoding design is formulated as a sparse reconstruction problem to minimize the Euclidean distance between the optimal full digital precoding matrix and the hybrid precoding matrix \cite{Sparse}. In particular, the array response vectors are spanned to generate the codebook of analog precoding matrix \cite{Sparse,ARVcs}. Based on the similar thought, the columns of analog precoding matrix are selected from the discrete Fourier transform matrix in \cite{DFT1,DFT2}. With the limitation of the codebook, the analog precoding design suffers a low degree of freedom. To approach the performance of the full digital precoder, the authors in \cite{optfhp} prove that it is sufficient for the number RF chains to be twice the number of data streams, and provide the closed-form expressions for the precoding design. When the RF chains is not enough, the decoupling design scheme is proposed in \cite{fuldecp1,fuldecp2}, where the analog precoder is first designed to harvest the large array gain, and the digital precoder is further obtained to manage the inter-user interference. Undoubtedly, the decoupling process shall lead to performance loss. Thus, how to improve fully-connected hybrid precoding (FHP) scheme with insufficient RF chains and achieve close performance to full digital precoding is an urgent problem to be solved.

In the sub-connected structure, each RF chain is connected to a specific subset of antennas. Since there is no overlap among antenna subsets, no RF adders are required. For single-user systems, the principle of manifold optimization and particle swarm optimization are respectively considered to develop two algorithms with different complexity in \cite{subpso}. The authors in \cite{subsnr} discuss the analog precoding design for high and low SNR conditions, respectively. The original multi-stream transmission problem is decomposed into several single-stream transmission problems with per-antenna power constraint. For multi-user systems, the codebook-based and the decoupling-based schemes can still be operated \cite{subcodbk1,subcodbk2,subcodbk3,subdecp}. Additionally, machine learning is proposed as a novel approach for sub-connected hybrid precoding (SHP) design \cite{subml1,subml2}. The authors in \cite{subml1} reformulate the beam selection problem for uplink precoding as a multiclass-classification problem which can be solved by the support vector machine algorithm. In \cite{subml2}, the analog precoder is realized by several switches and inverters. And an adaptive cross-entropy-based scheme is developed for the new architecture. Essentially, the new structure can be equivalently realized by one-bit quantized APSs. Consequently, it is common for recent study to impose extra constraints for sub-connected structure, such as codebook-based analog precoder and APSs with few quantization bits, which limits the freedom of design.

The numerical results in \cite{subdecp} illustrate that the fully-connected structure can provide better precoding performance than the sub-connected structure. Nevertheless, the implementation of the fully-connected structure in massive MIMO systems requires high hardware consumption due to the large demand of APSs and RF adders. To provide a tradeoff between two structures, the adaptively-connected hybrid precoding (AHP) scheme is adopted in \cite{AHPsu,AHPmu,AHPmu2,MYpimrc}. As shown in Fig. 1(c), the adaptively-connected structure is a kind of generalization for the sub-connected structure. The adaptive connection network provides flexible connection between RF chains and antennas, which means better precoding performance can be achieved with the similar hardware consumption as the sub-connected structure. Considering single-user scenarios, the authors in \cite{AHPsu} propose the connecting scheme based on maximizing the sum of the largest singular values of several subchannel matrices. For multi-user scenarios, the decoupling-based schemes are further revised in \cite{AHPmu,AHPmu2}, where the analog preconding matrix is designed to improve the users' average achievable rate with the thought of greed. So far, less research efforts have been invested in AHP. Especially for multi-user scenarios, the decoupling-based schemes ignore the relationship between analog and digital precoders and make less use of the flexibility in the adaptive connection network, which results in poor performance. Moreover, only a portion of the RF chains are effectively utilized in \cite{AHPmu,AHPmu2}, the number of which is equal to that of users.

\begin{figure*}[!t]
\captionsetup{belowskip=-10pt}
\centering
\subfigure[]{
\includegraphics[scale=0.5]{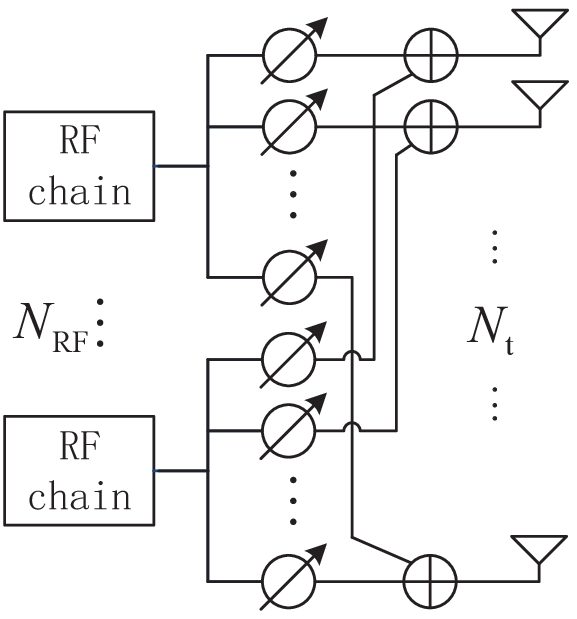}}
\subfigure[]{
\includegraphics[scale=0.5]{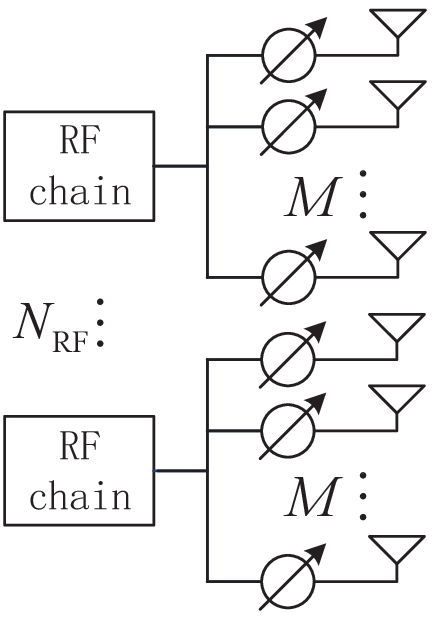}}
\subfigure[]{
\includegraphics[scale=0.5]{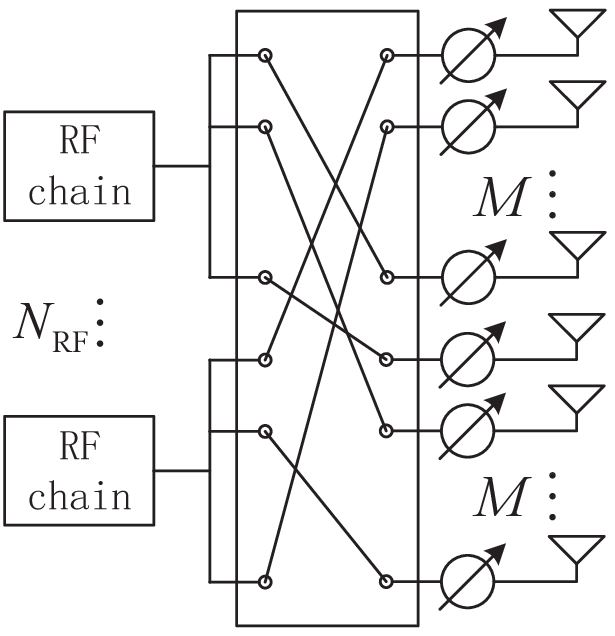}}
\subfigure[]{
\includegraphics[scale=0.5]{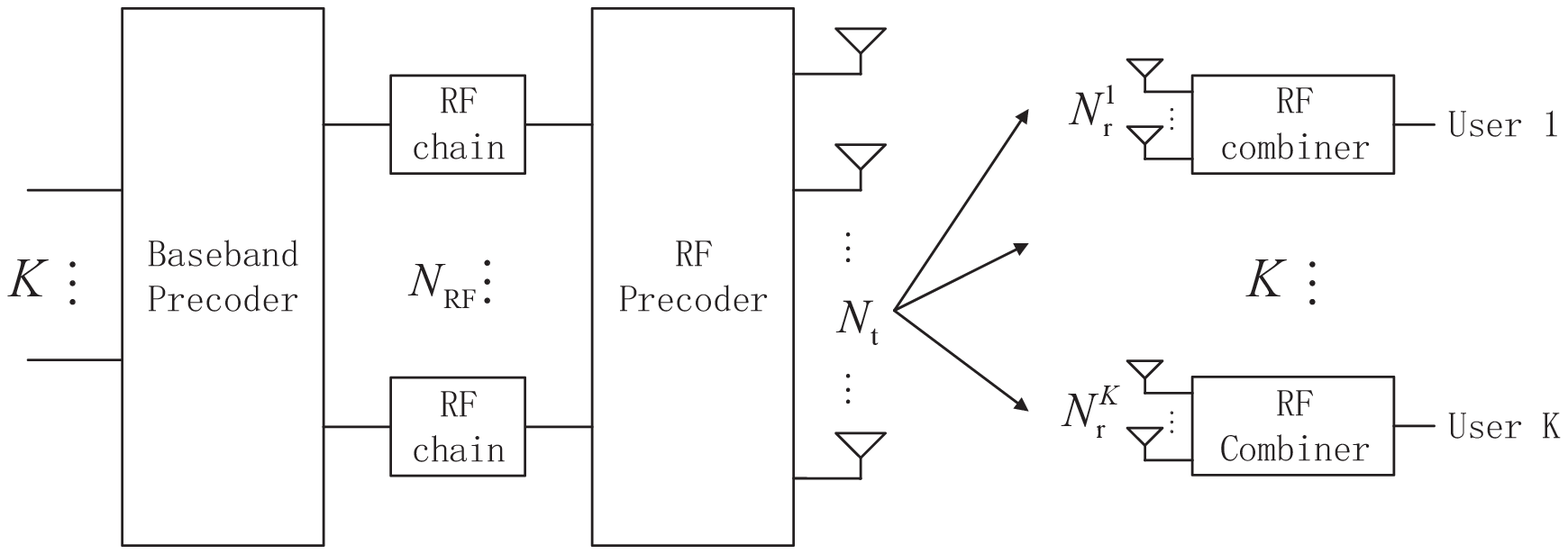}}
\caption{(a) Hybrid precoder with fully-connected structure; (b) Hybrid precoder with sub-connected structure; (c) Hybrid precoder with adaptively-connected structure; (d) A multi-user downlink
massive MIMO system with hybrid precoding.}
\label{fig.1} 
\end{figure*}

\subsection{Contributions}
In this paper, we propose the hybrid precoder design for multi-user massive MIMO systems in mmWave communication. Innovatively, we reformulate the FHP and AHP design problems as clustering problems and propose to solve them based on unsupervised learning methods. The main contributions are summarized as follows:

\begin{itemize}
 \item For FHP, we propose to reformulate the precoding problem as a clustering problem by minimizing the upper bound of the Euclidean distance between the optimal full digital precoder and the hybrid precoder. By exploring the relevance among RF chains in optimal hybrid procoder,  the hierarchical-agglomerative-clustering-based (HAC-based) FHP scheme is proposed with a novel defined distance function. Moreover, upper bound of the proposed scheme is analyzed.
 \item For AHP, we utilize the characteristic of the structure and simplify the precoding problem as a semi-unitary matrix factorization problem, which is equivalent to a clustering problem. By exploring the relevance among antennas at the BS, the modified-K-means-based (MKM-based) AHP scheme is proposed to make full use of the flexible connection.
 \item In the MKM-based AHP scheme, we propose to update the clustering centers with the alternating-optimization-based (AO-based) algorithm, where a specific initialization scheme is developed to reduce the computational complexity. In addition, we clarify that the proposed AO-based algorithm is feasible for SHP design.

\end{itemize}

Simulation results demonstrate the superiority of the proposed clustering-based precoding scheme. Specifically, the HAC-based FHP scheme provides close spectral efficiency to the full digital precoder with insufficient RF chains. The MKM-based AHP scheme contributes to high power efficiency. Including the AO-based SHP scheme, all the proposed schemes provide satisfying performance gain compared with the existing work.

\subsection{Organization}
The remainder of this paper is organized as follows. Section II introduces the system model, channel model, and the problem formulation for the precoding design. In Section III and IV, two clustering-based schemes are proposed for FHP and AHP, respectively. Particularly, the AO-based algorithm in Section IV-B is proposed as a suitable scheme for SHP design. Then, simulation results are presented in Section V to demonstrate the superior performance of proposed schemes. Finally, some conclusions are given in Section VI.


   \emph{Notations}: $a$, $\bf{a}$ and $\bf{A}$ denote a scalar, vector and matrix, respectively. For a given matrix $\mathbf A$, ${\mathbf A^T }$, ${\mathbf A^H }$, ${\mathbf A^{-1} }$, ${\mathbf A^{\dagger} }$, and ${\rm{r}}(\bf{A})$ denote its transpose, conjugate transpose, inverse, pseudo-inverse and rank, respectively. $\mathbf{A}(m,n)$, $\mathbf{A}(m,:)$ and $\mathbf{A}(:,n)$ denote the $(m,n)$-th entry, the $m$-th row and $n$-th column of $\mathbf{A}$, respectively. ${{\bf{A}}^{[:n]}}$ and ${{\bf{A}}^{[n:]}}$ denote the submatrices of ${\bf{A}}$ formed by the first $n$ columns and rows. The Frobenius norm and $\ell_2$ norm are noted by ${\left\|  \cdot  \right\|_F}$ and ${\left|  \cdot  \right|_2}$. ${{\bf{I}}_N}$ denotes the $N \times N$ identity matrix, while $\mathbf 0_{M\times N}$ denotes the $M\times N$ all-zero matrix. $\mathcal{CN}(\bm \alpha,\mathbf R)$ denotes the complex Gaussian distribution with mean $\bm \alpha$ and covariance $\mathbf R$. $\emptyset$ denotes the empty set. $\mathbb{E}\left [ \cdot \right ]$ denotes the expectation. The magnitude and phase of a complex scalar are denoted by ${\left|  \cdot  \right|}$ and ${\rm{arg}}\{ \cdot \}$. The SVD of $\mathbf A$ is in the form of  $\mathbf A = {\mathbf U}{\mathbf {\Sigma}}{\mathbf V}^H$, where $\bf{U}$ and $\bf{V}$ are left-singular and right-singular matrices, and ${{\bf{\Sigma }}}$ is a rectangular diagonal matrix with descending ordered singular values on the diagonal.

\section{System Model and Problem Formulation}

\subsection{System Model}
As shown in Fig. 1(d), we consider the downlink communication of a multi-user massive MIMO system with a hybrid precoder and combiners. The base station (BS) is equipped with ${N_{{\rm{RF}}}}$ RF chains and ${N_{\rm{t}}}$ transmit antennas to serve $K$ users. Based on the principles of low-cost and low-power consumption for mobile terminal design \cite{fuladd2}, we further assume that the $k$-th user is equipped with one RF chain and ${{N_{{\rm{r}},k}}}$ receive antennas. Thus, each user shall require only a single data stream from the BS. Due to the low dimensionality of digital precoding in the hybrid structure, we typically have $K \le {N_{{\rm{RF}}}} \ll {N_{\rm{t}}}$. Mathematically, the linear transmit precoded signal from the BS can be represented as
  \begin{equation}
  \mathbf{s} = \mathbf{F}_{{\rm{RF}}}  \mathbf{F}_{{\rm{BB}}}  \mathbf{x},
  \end{equation}
where $\mathbf{F}_{{\rm{RF}}}\in \mathbb{C}^{N_{{\rm{t}}} \times N_{{\rm{RF}}}}$ denotes the analog precoder in the RF domain, $\mathbf{F}_{\rm{BB}} \in \mathbb{C}^{N_{{\rm{RF}}} \times K}$ denotes the digital precoder in the baseband, and $\mathbf{x} \in \mathbb{C}^{K \times 1}$ denotes the transmit symbol vector. Without loss of generality, the average total transmit power of the BS is setted as $P$ with $\mathbf{x}$ satisfying $\mathbb{E}\left [ \mathbf{xx}^{H}\right ] = \frac{P}{K}\mathbf{I}_{K}$. Thus, the power constraint of the overall hybrid precoder can be given by
  \begin{equation}
  \left \| \mathbf{F}_{{\rm{RF}}}\mathbf{F}_{{\rm{BB}}} \right \|_{F}^{2}=K.
  \end{equation}

For simplicity, the block-fading channel model is adopted in this paper. At the $k$-th user, the received signal is further processed by own RF combiner, which can be expressed as
\begin{equation}
  {r_k} = {\bf{w}}_k^H{{\bf{H}}_k}{{\bf{F}}_{{\rm{RF}}}}{{\bf{F}}_{{\rm{BB}}}}{\bf{x}} + {\bf{w}}_k^H{{\bf{n}}_k},
  \end{equation}
where ${{\bf{w}}_k} \in {\mathbb{C}^{{N_{{\rm{r}},k}} \times 1}}$ denotes the RF combiner of the $k$-th user, ${{\bf{H}}_k} \in {\mathbb{C}^{{N_{{\rm{r}},k}} \times N_{\rm{t}}}}$ denotes the channel matrix between the BS and the $k$-th user, ${{\bf{n}}_k} \in {\mathbb{C}^{{N_{{\rm{r}},k}} \times 1}}$ denotes a complex Gaussian noise vector with each element obeying $\mathcal {CN}(0,\sigma _k^2)$ (assumed same for each user, i.e., $\sigma _k^2 = {\sigma ^2}, \forall k$).

The Saleh-Valenzuela model is commonly accepted to characterize the limited scattering feature of mmWave channel \cite{fuldecp1,fuladd,MYpimrc}, which is also adopted in this paper. The normalized channel for the $k$-th user consists of ${{N_{{\rm{c,}}k}}}$ scattering clusters, each of which is a sum of contributions of ${{N_{{\rm{p,}}k}}}$ propagation paths, which can be depicted as
 \begin{equation}
 {{\bf{H}}_k} = \sqrt {\frac{{{N_{\rm{t}}}{N_{{\rm{r}},k}}}}{{{N_{{\rm{c,}}k}}{N_{{\rm{p,}}k}}}}} \sum\limits_{c = 1}^{{N_{{\rm{c,}}k}}} {\sum\limits_{p = 1}^{{N_{{\rm{p,}}k}}} {{\beta _{c,p}}{{\bf{a}}_{\rm{r}}}(\theta _{c,p}^{\rm{r}},\phi _{c,p}^{\rm{r}}){\bf{a}}_{\rm{t}}^H(\theta _{c,p}^{\rm{t}},\phi _{c,p}^{\rm{t}})} } ,
 \end{equation}
where ${{\beta _{c,p}}} \sim {\cal C}{\cal N}(0,1)$ denotes the complex gain of the $p$-th path in the $c$-th cluster. In addition, ${{{\bf{a}}_{\rm{r}}}(\theta _{c,p}^{\rm{r}},\phi _{c,p}^{\rm{r}})}$ and ${{{\bf{a}}_{\rm{t}}}(\theta _{c,p}^{\rm{r}},\phi _{c,p}^{\rm{r}})}$ denote the normalized receive and transmit array response vectors corresponding to the azimuth (elevation) angle of arrival ${\theta _{c,p}^{\rm{r}}}$ (${\phi _{c,p}^{\rm{r}}}$) and departure ${\theta _{c,p}^{\rm{t}}}$ (${\phi _{c,p}^{\rm{t}}}$), respectively. Since the proposed algorithms in this paper are applicable for arbitrary antenna arrays, the uniform planar array (UPA) will be considered for the completeness of simulations. In the case of UPA, the $W\times V$-element array's response is variant in two angle domain, which can be expressed as
  \begin{equation}
\begin{split}
{{\bf{a}}_{{\rm{UPA}}}}(\theta ,\phi )= &\frac{1}{{\sqrt {WV} }}\left[ {1,...,{e^{j\frac{{2\pi d}}{\lambda }(w\sin (\theta )\sin (\phi ) + v\cos (\phi ))}},} \right.\\
&{\left. {...,{e^{j\frac{{2\pi d}}{\lambda }((W - 1)\sin (\theta )\sin (\phi ) + (V - 1)\cos (\phi ))}}} \right]^T},
\end{split}
 \end{equation}
where $0 \le w < W$ and $0 \le v < V$. As the basis of the precoding design in this paper, ${{\bf{H}}_k}$ is assumed known at the BS and the $k$-th users, i.e., the BS owns the global CSI, while each user only holds its own part \cite{Sparse,subml2,AHPmu}.

\subsection{Problem Formulation}

Since the RF combiner are realized by APSs which can only adjust the phases of signals, the entries in ${{\bf{w}}_k}$ satisfy the constant modulus constraint, which can be expressed as
 \begin{equation}
 \left| {{{\bf{w}}_k}(i)} \right| = 1/\sqrt {{N_{{\rm{r}},k}}} ,\forall i,
 \end{equation}
where $1/\sqrt {{N_{{\rm{r}},k}}}$ is the normalization parameter to satisfy $\left\| {{{\bf{w}}_k}} \right\|_F^2 = 1$. Moreover, the available phases of APSs are quantized in general due to the practical hardware constraint. Thus, a more strictly restricted RF combiner can be given by
 \begin{equation}
{{\bf{w}}_k}(i) \in \left\{ {\frac{1}{{\sqrt {{N_{{\rm{r}},k}}} }}{e^{j\frac{{2\pi q}}{{{2^Q}}}}}:q = 0,1,...,{2^Q} - 1} \right\},\forall i,
 \end{equation}
where $Q$ denotes the quantization bit number of APSs. Actually, once the unquantized ${{\bf{w}}_{k,{\rm{unq}}}}$ is obtained, the quantized one can be further given by\footnote{The similar operation is feasible for quantized RF precoder design.}
 \begin{equation}
 {{\bf{w}}_{k,{\rm{q}}}}(i) = \frac{1}{{\sqrt {{N_{{\rm{r}},k}}} }}{e^{j\frac{{2\pi \hat q}}{{{2^Q}}}}},
 \end{equation}
with
 \begin{equation}
 \hat q = \mathop {\arg \min }\limits_{q = 1,2,...,{2^Q} - 1} \left| {{{\bf{w}}_{k,{\rm{unq}}}}(i) - \frac{{{e^{j\frac{{2\pi q}}{{{2^Q}}}}}}}{{\sqrt {{N_{{\rm{r}},k}}} }}} \right|.
 \end{equation}
Obviously, a larger $Q$ contributes to higher quantization accuracy, which will be further clarified in Section V. In the following design, the unquantized RF combiner and precoder will be mainly discussed without loss of generality.

Since ${{\bf{w}}_k}$ is always a unit vector, the power of the Gaussian noise prossed by RF combiner still maintains at ${\sigma ^2}$ due to the unitary transformation. Recall the expression of the received signal (3), we obtain the signal-to-interference-plus-noise-ratio (SINR) of the $k$-th user as
  \begin{equation}
  {\rm{SIN}}{{\rm{R}}_k} = \frac{{P{{\left| {{\bf{w}}_k^H{{\bf{H}}_k}{{\bf{F}}_{{\rm{RF}}}}{{\bf{F}}_{{\rm{BB}}}}(:,k)} \right|}^2}}}{{K{\sigma ^2} + \sum\nolimits_{l \ne k} {P{{\left| {{\bf{w}}_k^H{{\bf{H}}_k}{{\bf{F}}_{{\rm{RF}}}}{{\bf{F}}_{{\rm{BB}}}}(:,l)} \right|}^2}} }}.
  \end{equation}

Since only a part of global CSI is available for each user, the inter-user interference is not visible at the user side. Thus, the multi-user system degenerates into the point-to-point system for each user, which motivates a selfish design for combiners to directly maximize the channel gain. According to \cite{Sparse}, the optimal selfish combiner for the $k$-th user is given by ${{\bf{w}}_k} = {{\bf{U}}_k}(:,1)$ with the SVD ${{\bf{U}}_k}{{\bf{\Sigma}}_k}{{\bf{V}}^H_k}={{\bf{H}}_k}$. Further based on the constraint (6), the RF combiner can be given by
  \begin{equation}
{{\bf{w}}_k}(i) = \frac{1}{{\sqrt {{N_{{\rm{r}},k}}} }}{e^{j\arg \{ {{\bf{U}}_k}(i,1)\} }}.
  \end{equation}

With the global CSI, all of the combining weights can be obtained at the BS. Then, the aim of the precoding design is to manage the inter-user interference and enhance system spectral efficiency, i.e.,
 \begin{equation}
{R = \sum\limits_{k = 1}^K {{{\log }_2}(1 + {\rm{SIN}}{{\rm{R}}_k})} }.
 \end{equation}
In this paper, we consider to minimize the Euclidean distance between optimal full digital precoding matrix ${{\bf{F}}_{{\rm{opt}}}}$ and the hybrid precoding matrix as follows
 \begin{equation}
\begin{array}{*{20}{c}}
{\mathop {\min }\limits_{{{\bf{F}}_{{\rm{RF}}}},{{\bf{F}}_{{\rm{BB}}}}} }&{\left\| {{{\bf{F}}_{{\rm{opt}}}} - {{\bf{F}}_{{\rm{RF}}}}{{\bf{F}}_{{\rm{BB}}}}} \right\|_F^2}\\
{{\rm{s}}{\rm{.t}}{\rm{.}}}&{\left\{ \begin{array}{l}
(2),\\
{\rm{constraints\;of\;}}{{\bf{F}}_{{\rm{RF}}}},
\end{array} \right.}
\end{array}
 \end{equation}
which has been proved as a sufficient precoding design scheme in \cite{Sparse}. In addition, the optimal full digital precoder is provided by the classical MF, ZF, and RZF method\footnote{The full digital precoder design is not the focus of this paper.} as
 \begin{equation}
 \begin{array}{*{20}{c}}
 {{{\bf{F}}_{{\rm{opt}}}} = \left\{ \begin{array}{l}
 \sqrt {{\gamma _{{\rm{MF}}}}} {\bf{H}}_{{\rm{eq}}}^H,\\
 \sqrt {{\gamma _{{\rm{ZF}}}}} {\bf{H}}_{{\rm{eq}}}^\dag, \\
 \sqrt {{\gamma _{{\rm{RZF}}}}} {\bf{H}}_{{\rm{eq}}}^H{({{\bf{H}}_{{\rm{eq}}}}{\bf{H}}_{{\rm{eq}}}^H + \beta {{\bf{I}}_K})^{ - 1},}
 \end{array} \right.}
 \end{array}
 \end{equation}
where $\sqrt {{\gamma _{{\rm{MF}}}}}$, $\sqrt {{\gamma _{{\rm{ZF}}}}}$, $\sqrt {{\gamma _{{\rm{RZF}}}}}$ denote normalization parameters to ensure the digital precoder satisfying the power constraint $\left\| {{{\bf{F}}_{{\rm{opt}}}}} \right\|_F^2 = K$, $\beta$ denotes the regularization parameter given by ${K{\sigma ^2}/P}$ \cite{ZFR}, ${\bf{H}}_{{\rm{eq}}}$ denotes the equivalent channel which can be expressed in a matrix form as
 \begin{equation}
{{\bf{H}}_{{\rm{eq}}}} = \left[ {\begin{array}{*{20}{c}}
{{\bf{w}}_1^H}&{}&{}&{}\\
{}&{{\bf{w}}_2^H}&{}&{}\\
{}&{}& \ddots &{}\\
{}&{}&{}&{{\bf{w}}_K^H}
\end{array}} \right]\left[ {\begin{array}{*{20}{c}}
{{{\bf{H}}_1}}\\
{{{\bf{H}}_2}}\\
 \vdots \\
{{{\bf{H}}_K}}
\end{array}} \right].
 \end{equation}
The impact of different digital precoding methods will be further clarified in Section V.

Although the objective function of (13) is in a simple form, the problem is still hard to solve due to the coupling of analog and digital precoding matrices as well as the unit modulus constraints of the analog precoder. Moreover, different connection structures between antennas and RF chains will bring various constraints to the problem. Based on the thought of clustering in unsupervised learning, the subsequent parts of this paper solve the optimization problem for FHP and AHP.

 \emph{Remark}: According to the expression of the received signal (3), the inter-user interference can be handled both at the combining and precoding sides. If the interference is handled at the combining side, all of the combiners should be jointly optimized at the BS, since the sufficient computing capacity and global CSI are not available to a single user terminal. Nevertheless, because the BS needs to feed back the combiner design to all users, even if a slight change occurs in the system, novel feedback will be required, which leads to unacceptable overhead in practice. Thus, the task of handling interference is completely put on the precoder at the BS, while the combiners can be simply given by a selfish design scheme in this paper.

\section{Clustering Based Fully-Connected Hybrid Precoding}
In this section, we first clarify the constraints of the RF precoder for FHP. Considering the optimal hybrid precoding scheme in \cite{optfhp}, we formulate the hybrid precoder design problem for the cases with insufficient RF chains, and further simplify the original problem by minimizing the upper bound of the objective function, which can be regarded as a clustering problem. Then, with the novel defined inter-cluster distance, the clustering and the center design parts are successively operated to develop the HAC-based FHP scheme. Essentially, the proposed scheme explores the relevance among RF chains in optimal hybrid precoder and merges the similar RF chains into a novel RF chain to efficiently reflect the original performance.

\subsection{Problem Derivation for FHP}
Since each pair of RF chain and antenna is connected via an APS and an RF adder in FHP, the constraint of RF precoder in (13) can be specified by
 \begin{equation}
 \left| {{{\bf{F}}_{{\rm{RF}}}}(i,n)} \right|=1.
 \end{equation}

It has been pointed out in \cite{optfhp} that $2N_{\rm{s}}$ (twice the number of transmit streams) RF chains are sufficient for FHP to realize the same performance as the optimal full digital precoder, i.e., ${{\bf{F}}_{{\rm{opt}}}} = {{\bf{F}}^{\star}_{{\rm{RF}}}}{{\bf{F}}^{\star}_{{\rm{BB}}}}$. With the single-stream transmission for each user, the original problem (13) is equivalent to\footnote{For the convenience of further derivation, the square of the Frobenius norm is neglected, which will not affect the solution.}
 \begin{equation}
 \begin{array}{*{20}{c}}
 {\mathop {\min }\limits_{{{\bf{F}}_{{\rm{RF}}}},{{\bf{F}}_{{\rm{BB}}}}} }&{{{\left\| {{\bf{F}}_{{\rm{RF}}}^ \star {\bf{F}}_{{\rm{BB}}}^ \star  - {{\bf{F}}_{{\rm{RF}}}}{{\bf{F}}_{{\rm{BB}}}}} \right\|}_F}}\\
 {{\rm{s}}.{\rm{t}}.}&{(2),(16)},
 \end{array}
 \end{equation}
where the number of columns in $\mathbf{F}_{{\rm{RF}}}^ \star$ and the number of rows in $\mathbf{F}_{{\rm{BB}}}^ \star$ are both $2K$. For the cases with insufficient RF chains, we have $K \leq N_{\rm{RF}} \leq 2K$. Due to the complexity of the problem (17), we consider to minimize the upper bound of the objective function. Based on ${\left\| {{\bf{A}} + {\bf{B}}} \right\|_F} \le {\left\| {\bf{A}} \right\|_F} + {\left\| {\bf{B}} \right\|_F}$, (17) can be simplified as
 \begin{equation}
\begin{array}{*{20}{c}}
{\mathop {\min }\limits_{{{\bf{F}}_{{\rm{RF}}}},{{\bf{F}}_{{\rm{BB}}}}} }&{\sum\limits_{c = 1}^{{N_{{\rm{RF}}}}} {{{\left\| {\sum\limits_{m \in {\Gamma _c}} {{\bf{F}}_{{\rm{HP}}}^ \star (m)}  - {{\bf{F}}_{{\rm{HP}}}}(c)} \right\|}_F}} }\\
{{\rm{s}}.{\rm{t}}.}&{\left\{ {\begin{array}{*{20}{l}}
{(2),(16),}\\
\begin{array}{l}
{\Gamma _c} \cap {\Gamma _d} = \emptyset ,\forall c \ne d,\\
\bigcup\limits_{c = 1}^{{N_{{\rm{RF}}}}} {{\Gamma _c}}  = \left\{ {1,2,...,2K} \right\},
\end{array}
\end{array}} \right.}
\end{array}
 \end{equation}
with the denotations ${{\bf{F}}_{{\rm{HP}}}^ \star (m)}={{\bf{F}}_{{\rm{RF}}}^ \star (:,m)}{{\bf{F}}_{{\rm{BB}}}^ \star (m,:)}$ and ${{\bf{F}}_{{\rm{HP}}}}(c) = {{\bf{F}}_{{\rm{RF}}}}(:,c){{\bf{F}}_{{\rm{BB}}}}(c,:)$. With such form of the objective function, (18) can be regarded as an atypical clustering problem. Specifically, the set of data samples is composed of $2K$ matrices with ${{\bf{F}}_{{\rm{HP}}}^ \star (m)}$ denoting each of them. The data samples are expected to be divided into ${{N_{{\rm{RF}}}}}$ clusters. For the $c$-th cluster, ${{\bf{F}}_{{\rm{HP}}}}(c)$ denotes the clustering center, while ${{\Gamma _c}}$ denotes the index set of members. In terms of data set composition, each RF chain in the optimal hybrid precoder generates a data sample, which implies the design of fully-connected structure is essentially the clustering of RF chains.

Based on (18), the FHP scheme can be developed with the successive operation of the clustering part and the center design part. The clustering part is to determine ${{N_{{\rm{RF}}}}}$ member index sets corresponding to expected clusters. For different ${{N_{{\rm{RF}}}}}$ cases, the HAC method is feasible for this part, where the key point is to define the distance between clusters. Then, in the center design part, each cluster will generate its own clustering center which is essentially corresponding to a certain component of the expected hybrid precoding matrix.

\subsection{Clustering and Center Design}

With each initial cluster formed by a single data sample, the key approach of the clustering part in HAC method is to find a pair of clusters with the smallest distance and merge them into a new cluster \cite{HAC}. For the fairness among RF chains, the distance between the $c$-th and the $d$-th cluster is given by the mean distance between members of each cluster as
 \begin{equation}
{\boldsymbol{\mathfrak{D}}}(c,d) = \frac{1}{{\left| {{\Gamma _c}} \right|\left| {{\Gamma _d}} \right|}}\sum\limits_{m \in {\Gamma _c}} {\sum\limits_{n \in {\Gamma _d}} {{\cal D}_{\rm{F}}({\bf{F}}_{{\rm{HP}}}^ \star (m),{\bf{F}}_{{\rm{HP}}}^ \star (n))} },
 \end{equation}
where ${\left| {{\Gamma}} \right|}$ denotes the total number of the members in set ${{\Gamma}}$, ${\cal D}_{\rm{F}}(\bf{A},\bf{B})$ denotes the inter-sample distance between $\bf{A}$ and $\bf{B}$. Generally, since the clustering center can be regarded as a virtual sample, the inter-sample distance can be similarly given based on the intra-cluster distance\footnote{For the typical intra-cluster distance given by the average distance between samples and the center $\frac{1}{{\left| {{\Gamma _c}} \right|}}\sum\nolimits_{m \in {\Gamma _c}} {{{\left\| {{\bf{F}}_{{\rm{HP}}}^ \star (m) - {{\bf{F}}_{{\rm{HP}}}}(c)} \right\|}_F}} $, the inter-sample distance can be simply given by ${\cal D}_{\rm{F}}({\bf{A}},{\bf{B}}) = {\left\| {{\bf{A}} - {\bf{B}}} \right\|_F}$.}. However, the intra-cluster distance in (19) is given by ${\left\| {\sum\nolimits_{m \in {\Gamma _c}} {{\bf{F}}_{{\rm{HP}}}^ \star (m)}  - {{\bf{F}}_{{\rm{HP}}}}(c)} \right\|_F}$, which is not standard. Thus, we consider to redefine the inter-sample distance function based on the following proposition.
\newtheorem*{prop}{Proposition 1}
\begin{prop}
Under the condition that the column space of ${\bf{F}}_{\rm{RF}}$ keeps unchanged, arbitrarily adjusted ${\bf{F}}_{\rm{RF}}$ will make no difference to hybrid precoding performance if ${\bf{F}}_{\rm{BB}}$ is appropriately adjusted.
\end{prop}
\begin{proof}
Please refer to the APPENDIX A.
\end{proof}

According to \textbf{Proposition 1}, the maintenance of the column space of ${\bf{F}}_{\rm{RF}}$ is the key point for precoding design. In (18), the column space of each sample is corresponding to a specific column of ${{\bf{F}}_{{\rm{RF}}}^ \star }$. To reduce the loss of the column space in the process of clustering, the samples with similar column space should be merged. Accordingly, we provide a reasonable distance definition for ${{\bf{F}}_{{\rm{HP}}}^ \star (m)}$ and ${{\bf{F}}_{{\rm{HP}}}^ \star (n)}$ with the inner product of ${{\bf{F}}_{{\rm{RF}}}^ \star (:,m)}$ and ${{\bf{F}}_{{\rm{RF}}}^ \star (:,n)}$, i.e.,
 \begin{equation}
{\cal D}_{\rm{F}}({\bf{F}}_{{\rm{HP}}}^ \star (m),{\bf{F}}_{{\rm{HP}}}^ \star (n)) = {\left| {{\bf{F}}_{{\rm{RF}}}^ \star {{(:,m)}^H}{\bf{F}}_{{\rm{RF}}}^ \star (:,n)} \right|^{ - 1}},
 \end{equation}
where a large inner product means a high similarity in column space and a small distance between samples.

With the definition of the distance for clusters and samples, the clustering part is operable to determine the membership of each cluster, i.e., $\forall{\Gamma _c}$ in (18) is determined. Then, the problem (18) can be decomposed into ${{N}}_{\rm{RF}}$ subproblems, where the center of the $c$-th cluster is obtained to generate a certain component of hybrid precoding matrix by
 \begin{equation}
\begin{array}{*{20}{c}}
{\mathop {\min }\limits_{{{\bf{F}}_{{\rm{RF}}}}(:,c)} }&{{{\left\| {\sum\limits_{m \in {\Gamma _c}} {{\bf{F}}_{{\rm{HP}}}^ \star (m)}  - {{\bf{F}}_{{\rm{HP}}}}(c)} \right\|}_F}}\\
{{\rm{s}}.{\rm{t}}.}&{(16)}.
\end{array}
 \end{equation}
Due to the decomposition, the power constraint (2) is temporarily neglected in each subproblem, which slightly brings about negative impacts \cite{AMmethod}. Essentially, each subproblem provides the design for a certain RF chain. Without the modulus constraint (16), the optimal solution can be given by ${\bf{F}}_{{\rm{RF}}}(:,c) = {{\bf{U}}_{{c}}}(:,1)$, where ${{\bf{U}}_c}{{\bf{\Sigma }}_c}{\bf{V}}_c^H = \sum\nolimits_{m \in {\Gamma _c}} {{\bf{F}}_{{\rm{HP}}}^ \star (m)} $. Therefore, a near optimal solution can be given by
  \begin{equation}
  {\bf{F}}_{{\rm{RF}}}(i,c) = {e^{j\arg \{ {{\bf{U}}_c}(i,1)\} }}.
  \end{equation}
Owing to the inter-sample distance definition, the samples in the same cluster tend to have similar column space, which means more power of $\sum\nolimits_{m \in {\Gamma _c}} {{\bf{F}}_{{\rm{HP}}}^ \star (m)} $ is converged in the maximum singular value, and contributes to smaller approximation error for problem (21).

\vspace{-0.3em} 
\subsection{Proposed HAC-based FHP Scheme}

 \begin{algorithm}[!t]
\caption{Proposed HAC-based FHP algorithm}
\begin{algorithmic}[1]
\REQUIRE Optimal full digital precoder ${\bf{F}}_{\rm{opt}}$, the number of RF chains $N_{\rm{RF}}$.
\ENSURE ${\bf{F}}_{\rm{RF}}$, ${\bf{F}}_{\rm{BB}}$.
\STATE Obtain the ideal fully-connected hybrid precoder design ${{\bf{F}}_{{\rm{opt}}}} = {{\bf{F}}^{\star}_{{\rm{RF}}}}{{\bf{F}}^{\star}_{{\rm{BB}}}}$ based on the expressions in \cite{optfhp}.
\STATE Initialize the number of clusters by $N_{\rm{cl}}=2K$.
\FOR{ $c = 1$ to $N_{\rm{cl}}$ }
\STATE Initialize the $c$-th cluster with the member ${{\bf{F}}^{\star}_{{\rm{HP}}}}(c) = {{\bf{F}}^{\star}_{{\rm{RF}}}}(:,c){{\bf{F}}^{\star}_{{\rm{BB}}}}(c,:)$ and the member index set ${{\Gamma _c}} = \left\{ c \right\}$.
\ENDFOR
\STATE For $\forall n \ne m$, calculate the distance between ${{\bf{F}}^{\star}_{{\rm{HP}}}}(n)$ and ${{\bf{F}}^{\star}_{{\rm{HP}}}}(m)$ based on (20).
\WHILE {$N_{\rm{cl}}>N_{\rm{RF}}$}
\STATE Calculate inter-cluster distance ${\boldsymbol{\mathfrak{D}}}(c,d), \forall c \ne d$ based on (19).
\STATE Find the pair of clusters with the smallest distance, i.e., $\left\langle {{\Gamma _{{c^ \star }}},{\Gamma _{{d^ \star }}}} \right\rangle  = \mathop {\arg \min }\limits_{{\Gamma _c},{\Gamma _d}} {\boldsymbol{\mathfrak{D}}}(c,d)$.
\STATE Merge the nearest clusters ${\Gamma _{{c^ \star }}} = {\Gamma _{{c^ \star }}} \cup {\Gamma _{{d^ \star }}},{\Gamma _{{d^ \star }}} = \emptyset $.
\STATE Update the number of current clusters by $N_{\rm{cl}}=N_{\rm{cl}}-1$, and renumber the non-empty set ${\Gamma}$ from $1$ to $N_{\rm{cl}}$.
\ENDWHILE
\STATE For each ${{\Gamma _c}}$, generate ${\bf{F}}_{\rm{RF}}(:,c)$ based on (22).
\STATE Obtain ${\bf{F}}_{\rm{BB}}$ by LS method, or the classical digital precoding methods (15) with the effective baseband channel.
\end{algorithmic}
\end{algorithm}


As a summary, the HAC-based FHP design is shown in Algorithm 1. Initially, each sample is regarded as a single cluster. From step 7 to step 12, the pair of clusters with minimum distance are merged in each loop, until the number of cluster is reduced to $N_{\rm{RF}}$. With the determined clusters, the RF precoder is obtained based on (22). Finally, the baseband precoder can be obtained by two alternative schemes. One scheme is based on least-square (LS) method, i.e., ${{\bf{F}}_{{\rm{BB}}}} = \sqrt {{\gamma _{{\rm{LS}}}}} {\bf{F}}_{{\rm{RF}}}^\dag {{\bf{F}}_{{\rm{opt}}}}$, where $\sqrt{\gamma _{{\rm{LS}}}}$ denotes the normalization parameter to satisfy the precoding power constraint. Obviously, the performance of this scheme mainly depends on the similarity between the column space of ${\bf{F}}_{{\rm{RF}}}$ and ${{\bf{F}}_{{\rm{opt}}}}$. The other scheme is based on classical digital precoding methods. Similar to (14), the second scheme is operated with the effective baseband channel ${{\bf{H}}_{{\rm{BB}}}} = {{\bf{H}}_{\rm{eq}}}{{\bf{F}}_{\rm{RF}}}$, which abates the influence of column space similarity to a certain extent. In terms of the upper bound, we provide the comparison of two alternative schemes in the following theorem. Detail discusses for the selection of these two schemes will be presented in Section V.
\newtheorem*{theorem}{Theorem 1}
\begin{theorem}
With ${\bf{F}}_{{\rm{opt}}} = \sqrt {{\gamma _{{\rm{ZF}}}}} {\bf{H}}_{{\rm{eq}}}^\dag$ as the input, the upper bounds of the system spectral efficiency for the HAC-based FHP design with the LS and ZF refinement schemes are respectively given by
 \begin{equation}
\begin{array}{l}
{R_{{\rm{LS}}}} = K{\log _2}(1 + \frac{P}{{{\sigma ^2}\left\| {{\bf{H}}_{{\rm{eq}}}^\dag } \right\|_F^2}}),\\
{R_{{\rm{ZF}}}} = K{\log _2}(1 + \frac{P}{{{\sigma ^2}\left\| {{{\bf{F}}_{{\rm{RF}}}}{\bf{H}}_{{\rm{BB}}}^\dag } \right\|_F^2}}).
\end{array}
 \end{equation}
And $R_{\rm{ZF}} \leq R_{\rm{LS}}$ is always satisfied.
\end{theorem}
\begin{proof}
Please refer to the APPENDIX B.
\end{proof}

In the proposed HAC-based FHP scheme, the optimal hybrid procoding is decomposed into $2K$ components, with each RF chain corresponding to one of them. In order to reduce the performance loss caused by insufficient RF chains, the similar components in the optimal hybrid precoder are clustered based on the relevance among RF chains. And the clustered components are represented by a clustering center corresponding to a new RF chain design, which can efficiently reflect the performance of the overall cluster.

The HAC-based FHP algorithm is obviously convergent, since two clusters are merged in each iteration until $N_{\rm{RF}}$ clusters are obtained. The computational complexity of the algorithm is mainly caused by step 6 and step 13. In step 6, we need to calculate $K(2K-1)$ inner products with the complexity of $o(N_{\rm{t}})$ for each, which totals $o(K^2N_{\rm{t}})$. In step 13, we need to calculate SVDs for $N_{\rm{RF}}$ matrices in the size of $N_{\rm{t}} \times K$, which totals $o({N_{\rm{RF}}}(KN^2_{\rm{t}}+K^3))$ \cite{SVD}. Since $N_{\rm{t}} \gg K$ in massive MIMO systems, the overall complexity is $o({N_{\rm{RF}}}KN^2_{\rm{t}})$.





  \emph{Remark}: Actually, for the problem (18), the K-means algorithm is also operable \cite{KM}. However, Due to the non-convexity of the problem and the fact that the number of expected clusters is even more than half the number of samples, the performance of K-means algorithm is greatly affected by the initial value. Thus, it is not advisable to solve (18) with the K-means algorithm.

\section{Clustering Based Adaptively-Connected Hybrid Precoding}
In this section, we first analyze the constraints of the RF precoder for AHP, and reformulate the hybrid precoder design as a clustering problem. Then, we propose the MKM-based AHP scheme with the iterative refinement of the clustering and the center updating. Particularly, the clustering centers are updated by the AO-based algorithm with a specific initialization scheme to reduce the computational complexity, which is also capable to provide SHP design independently. Essentially, the MKM-based AHP scheme aims at exploring the relevance among antennas and merging the similar antennas into the same cluster to make full use of the flexible connection.

\subsection{Problem Derivation for AHP}
The major characteristic of AHP is that each RF chain is connected with a flexible subset of the antennas. Accordingly, the hardware constraint in (13) can be expressed as
  \begin{subequations}
  \begin{equation}
  \sum\nolimits_{n = 1}^{{N_{{\rm{RF}}}}} {\left| {{{\bf{F}}_{{\rm{RF}}}}(i,n)} \right|}  = 1,\forall i,
  \end{equation}
  \begin{equation}
  \sum\nolimits_{i = 1}^{{N_{\rm{t}}}} {\left| {{{\bf{F}}_{{\rm{RF}}}}(i,n)} \right|}  = M,\forall n.
  \end{equation}
  \end{subequations}
The constraint (24a) results from that each antenna is supported by an individual RF chain, while the constraint (24b) results from that each RF chain supports $M=N_{\rm{t}}/N_{\rm{RF}}$ (commonly assumed as an integer \cite{AHPmu,AHPmu2,MYpimrc}) antennas. $\mathbf{F}_{{\rm{RF}}}(i,n) \neq 0$ represents that the $i$-th antenna is supported by the $n$-th RF chain via an APS. Owing to the adaptive connection, the position of nonzero entries in analog precoding matrix is flexible.

Utilizing the SVD method, the objective function in (13) can be written as
  \begin{equation}
\begin{array}{*{20}{c}}
{\mathop {\min }\limits_{{{\bf{F}}_{{\rm{RF}}}},{{\bf{F}}_{{\rm{BB}}}}} }&{\left\| {{{\bf{U}}_{{\rm{opt}}}}{{\bf{\Sigma }}_{{\rm{opt}}}}{\bf{V}}_{{\rm{opt}}}^H - {{\bf{F}}_{{\rm{RF}}}}{{\bf{U}}_{{\rm{BB}}}}{{\bf{\Sigma }}_{{\rm{BB}}}}{\bf{V}}_{{\rm{BB}}}^H} \right\|_F^2},
\end{array}
  \end{equation}
where ${{\bf{\Sigma }}_{{\rm{opt}}}}\in {\mathbb{C}^{{N_{{\rm{t}}}} \times K}}$ and ${{\bf{\Sigma }}_{{\rm{BB}}}}\in {\mathbb{C}^{{N_{{\rm{RF}}}} \times K}}$ denote the singular value matrix of ${{\bf{F}}_{{\rm{opt}}}}$ and ${{\bf{F}}_{{\rm{BB}}}}$, respectively. Since the special characteristic of AHP keeps the equation ${\bf{F}}_{{\rm{RF}}}^H{{\bf{F}}_{{\rm{RF}}}} = M{{\bf{I}}_{{N_{{\rm{RF}}}}}}$ always established, the singular values of hybrid procoding matrix are mainly determined by the baseband precoder. Thus, the singular value matrix and the right-singular matrix of ${{\bf{F}}_{{\rm{BB}}}}$ can be unified by ${{\bf{F}}_{{\rm{opt}}}}$ as
  \begin{eqnarray}
   \begin{aligned}
      {\bf{\Sigma }}_{{\rm{BB}}}^{[K:]} &= {\bf{\Sigma }}_{{\rm{opt}}}^{[K:]}/\sqrt M ,\\
      {{\bf{V}}_{{\rm{BB}}}} &= {{\bf{V}}_{\rm{opt}}}.
   \end{aligned}
  \end{eqnarray}
Since ${{\bf{F}}_{{\rm{opt}}}}$ has no more than $K$ singular values, the problem (13) can be simplified as a semi-unitary matrix factorization problem, i.e.,
  \begin{equation}
\begin{array}{*{20}{c}}
{\mathop {\min }\limits_{{{\bf{F}}_{{\rm{RF}}}},{{\bf{U}}_{{\rm{BB}}}}} }&{\left\| {\sqrt M {\bf{U}}_{{\rm{opt}}}^{[:K]} - {{\bf{F}}_{{\rm{RF}}}}{\bf{U}}_{{\rm{BB}}}^{[:K]}} \right\|_F^2}\\
{{\rm{s}}.{\rm{t}}.}&{\left\{ {\begin{array}{*{20}{l}}
{(24\rm{a}),(24\rm{b})}\\
{{{\bf{U}}^H_{{\rm{BB}}}}{{\bf{U}}_{{\rm{BB}}}} = {{\bf{U}}_{{\rm{BB}}}}{{\bf{U}}^H_{{\rm{BB}}}} = {{\bf{I}}_{{N_{{\rm{RF}}}}}}.}
\end{array}} \right.}
\end{array}
  \end{equation}

With the constraint (24a), the rows in ${{\bf{F}}_{{\rm{RF}}}}{\bf{U}}_{{\rm{BB}}}^{[:K]}$ are rotated by the rows in ${\bf{U}}_{{\rm{BB}}}^{[:K]}$ with nonzero ${{\bf{F}}_{{\rm{RF}}}}(m,n)$. Further according to $\left\| {\bf{A}} \right\|_F^2 = \sum\nolimits_m {\left| {{\bf{A}}(m,:)} \right|_2^2} $, the problem (27) is
equivalent to minimize
  \begin{equation}
\sum\limits_{n = 1}^{{N_{{\rm{RF}}}}} {\sum\limits_{m \in {\Gamma _n}} {\left| {\sqrt M {\bf{U}}_{{\rm{opt}}}^{[:K]}(m,:) - {{\bf{F}}_{{\rm{RF}}}}(m,n){\bf{U}}_{{\rm{BB}}}^{[:K]}(n,:)} \right|_2^2} },
  \end{equation}
with the constraint (24) rewritten by
  \begin{equation}
\begin{array}{c}
{\Gamma _n} \cap {\Gamma _m} = \emptyset ,\forall n \ne m,\\
\bigcup\limits_{n=1}^{{N_{{\rm{RF}}}}} {{\Gamma _n}}  = \{ 1,2,...,{N_{\rm{t}}}\} ,\\
\left| {{\Gamma _n}} \right| = M,\forall n,
\end{array}
  \end{equation}
where ${{\Gamma _n}}$ is the index set of nonzero entries in ${{{\bf{F}}_{{\rm{RF}}}}(:,n)}$. Intuitively, the equivalent problem can be regarded as an atypical clustering problem. Specifically, the set of data samples is consisted by the rows of $\sqrt M {\bf{U}}_{{\rm{opt}}}^{[:K]}$ which are expected to be clustered into $N_{\rm{RF}}$ clusters. For the $n$-th cluster, ${\bf{U}}_{{\rm{BB}}}^{[:K]}(n,:)$ denotes the clustering center, while ${{\Gamma _n}}$ denotes the index set of members.

The clustering problem (28) can be solved based on K-means algorithm. The key approach of the K-means algorithm is the iterative refinement of the clustering and the center updating with the thought of greed \cite{KM}. However, due to the extra unitary constraint and (29), the classical K-means algorithm is no longer applicable, which means modification is imperative.

\subsection{Clustering and AO-based Center Updating}
With fixed clustering centers ${\bf{U}}_{{\rm{BB}}}^{[:K]}$, the clustering part of the iterative refinement is to assign each sample to the nearest cluster. The key of this part is to define the distance between samples and clustering centers.

Since clustering centers can be arbitrarily rotated by the nonzero entries in ${{\bf{F}}_{{\rm{RF}}}}$, the Euclidean distance cannot be directly used for distance definition. Further considering that if $\sqrt M {\bf{U}}_{{\rm{opt}}}^{[:K]}(m,:)$ is assigned to the cluster centered by ${\bf{U}}_{{\rm{BB}}}^{[:K]}(n,:)$, ${{\bf{F}}_{{\rm{RF}}}}(m,n)$ should be nonzero (i.e., $m \in {\Gamma _n}$) and determined by the following problem
  \begin{equation}
\begin{array}{*{20}{c}}
{\mathop {\min }\limits_{{{\bf{F}}_{{\rm{RF}}}}(m,n)} }&{\left| {\sqrt M {\bf{U}}_{{\rm{opt}}}^{[:K]}(m,:) - {{\bf{F}}_{{\rm{RF}}}}(m,n){\bf{U}}_{{\rm{BB}}}^{[:K]}(n,:)} \right|_2^2}\\
{{\rm{s}}.{\rm{t}}.}&{\left| {{{\bf{F}}_{{\rm{RF}}}}(m,n)} \right| = 1.}
\end{array}
  \end{equation}
And there exists a closed-form solution for (30) given by
  \begin{equation}
{{\bf{F}}_{{\rm{RF}}}}(m,n) = {e^{j\arg \{ {\bf{U}}_{{\rm{opt}}}^{[:K]}(m,:){\bf{U}}_{{\rm{BB}}}^{[:K]}{{(n,:)}^H }\} }}.
  \end{equation}
Thus, we define the distance between sample vector ${{\bf{s}}}$ and clustering center ${{\bf{c}}}$ as
  \begin{equation}
{\cal D}_{\rm{A}}({\bf{s}},{\bf{c}}) = \left| {{\bf{s}} - {e^{j\arg \{{\bf{s}}{{\bf{c}}^H}\}}}{\bf{c}}} \right|_2^2.
  \end{equation}


The other part of the iterative refinement is to update clustering centers with fixed ${{\Gamma _n}}$. The fixed nonzero positions of ${{\bf{F}}_{{\rm{RF}}}}$ mean the connection between antennas and RF chains is determined, which degenerates AHP into SHP. Mathematically, the objective function in (27) can be rewritten as
  \begin{equation}
\begin{array}{*{20}{c}}
{\mathop {\min }\limits_{{{\bf{F}}_{{\rm{RF}}}},{{\bf{U}}_{{\rm{BB}}}}} }&{\left\| {\sqrt M {\bf{RU}}_{{\rm{opt}}}^{[:K]} - {\bf{R}}{{\bf{F}}_{{\rm{RF}}}}{\bf{U}}_{{\rm{BB}}}^{[:K]}} \right\|_F^2,}
\end{array}
  \end{equation}
where ${\bf{R}}$ is obtained by rearranging the rows of ${{\bf{I}}_{{N_{\rm{t}}}}}$ to make the equivalent analog precoding ${\widetilde{{\bf{F}}}_{{\rm{RF}}}}={{\bf{R}}{{\bf{F}}_{{\rm{RF}}}}}$ a block dialog matrix, i.e.,
  \begin{equation}
{\widetilde{{\bf{F}}}_{{\rm{RF}}}} = \left[ {\begin{array}{*{20}{c}}
{{\widetilde{{\bf{f}}}_1}}& \cdots &{\bf{0}}\\
 \vdots & \ddots & \vdots \\
{\bf{0}}& \cdots &{{\widetilde{{\bf{f}}}_{{N_{{\rm{RF}}}}}}}
\end{array}} \right],
  \end{equation}
where ${{\widetilde{\bf{f}}}_n} \in {\mathbb{C}^{M \times 1}}$. Actually, the block dialog analog precoding matrix is the key characteristic of SHP.

 \begin{algorithm}[!t]
\caption{Proposed AO-based clustering center updating algorithm}
\begin{algorithmic}[1]
\REQUIRE The left-singular matrix of the full digital precoder ${\bf{U}}_{\rm{opt}}$, the index set of nonzero entries ${{\Gamma _n}}$.
\ENSURE ${{{\bf{F}}}_{\rm{RF}}}$, ${\bf{U}}_{\rm{BB}}$.
\STATE Calculate ${\bf{R}}$ based on ${{\Gamma _n}}$ to make $\widetilde{{\bf{F}}}_{{\rm{RF}}}$ block dialog.
\STATE Initialize $\widetilde{{\bf{F}}}_{{\rm{RF}}}$ based on (40).
\STATE Initialize the value of objective function as infinite $v_{0}=+\infty$, and set $k=1$.
\REPEAT
\STATE Calculate the SVD $\widetilde{{\bf{F}}}_{{\rm{RF}}}^H \widetilde{{\bf{U}}}_{{\rm{opt}}}^{[:K]}={{\bf{U}}_{\rm{B}}}{{\bf{\Sigma }}_{\rm{B}}}{\bf{V}}_{\rm{B}}^H$.
\STATE Obtain ${\bf{U}}_{{\rm{BB}}}^{[:K]}$ based on (38).
\STATE Update the equivalent analog precoding matrix ${\widetilde{{\bf{F}}}_{\rm{RF}}}$ based on ${\widetilde{{\bf{f}}}_n}(m) = {e^{j\arg \{ \widetilde{{\bf{U}}}_{{\rm{opt}}}^{[:K]}(t,:){\bf{U}}_{{\rm{BB}}}^{[:K]}{{(n,:)}^H }\} }}$.
\STATE $k=k+1$.
\STATE $v_{k}= {\left\| {\sqrt M \widetilde{{\bf{U}}}_{{\rm{opt}}}^{[:K]} - {\widetilde{{\bf{F}}}_{{\rm{RF}}}}{\bf{U}}_{{\rm{BB}}}^{[:K]}} \right\|_F^2} $.
\UNTIL {${v_{k - 1}} - {v_k} < \varepsilon $}
\STATE ${{{\bf{F}}}_{\rm{RF}}} = {\bf{R}}^{-1}{\widetilde{{\bf{F}}}_{\rm{RF}}}$.
\end{algorithmic}
\end{algorithm}

\vspace{2mm}
Since (33) is a joint optimization problem with non-convex constraints, we propose to utilize the AO method to find the near optimal solution. Specifically, one part of the AO-based method is to optimize ${\widetilde{{\bf{F}}}_{{\rm{RF}}}}$ with fixed ${{\bf{U}}_{{\rm{BB}}}}$ based on
  \begin{equation}
  \begin{array}{*{20}{c}}
  {\mathop {\min }\limits_{{\widetilde{{\bf{F}}}_{{\rm{RF}}}}} }&{\left\| {\sqrt M \widetilde{{\bf{U}}}_{{\rm{opt}}}^{[:K]} - {\widetilde{{\bf{F}}}_{{\rm{RF}}}}{\bf{U}}_{{\rm{BB}}}^{[:K]}} \right\|_F^2}\\
  {{\rm{s}}.{\rm{t}}.}&{\left| {{{\widetilde{\bf{f}}}_n}(m)} \right| = 1},
  \end{array}
  \end{equation}
where $\widetilde{{\bf{U}}}_{{\rm{opt}}}^{[:K]} = {{\bf{R}}}{{\bf{U}}}_{{\rm{opt}}}^{[:K]}$. Observing that ${\widetilde{{\bf{F}}}_{{\rm{RF}}}}{\bf{U}}_{{\rm{BB}}}^{[:K]}$ is formed from the rotated rows of ${\bf{U}}_{{\rm{BB}}}^{[:K]}$, we decompose the approximation into $N_{\rm{t}}$ subproblems as
  \begin{equation}
\begin{array}{*{20}{c}}
{\mathop {\min }\limits_{{\widetilde{{\bf{f}}}_n}(m)} }&{\left\| {\sqrt M \widetilde{{\bf{U}}}_{{\rm{opt}}}^{[:K]}(t,:) - {\widetilde{{\bf{f}}}_n}(m){\bf{U}}_{{\rm{BB}}}^{[:K]}(n,:)} \right\|_F^2}\\
{{\rm{s}}.{\rm{t}}.}&{\left| {{\widetilde{{\bf{f}}}_n}(m)} \right| = 1,}
\end{array}
  \end{equation}
where $t={Mn - M + m}$. The problem (36) is in the same form as (30), which can be
solved by ${\widetilde{{\bf{f}}}_n}(m) = {e^{j\arg \{ \widetilde{{\bf{U}}}_{{\rm{opt}}}^{[:K]}(t,:){\bf{U}}_{{\rm{BB}}}^{[:K]}{{(n,:)}^H }\} }}$.

The other part of the AO method is to optimize ${\bf{U}}_{{\rm{BB}}}$ with fixed $\widetilde{{\bf{F}}}_{{\rm{RF}}}$ based on
  \begin{equation}
  \begin{array}{*{20}{c}}
  {\mathop {\min }\limits_{{{\bf{U}}_{{\rm{BB}}}}} }&{\left\| {\sqrt M \widetilde{{\bf{U}}}_{{\rm{opt}}}^{[:K]} - {\widetilde{{\bf{F}}}_{{\rm{RF}}}}{\bf{U}}_{{\rm{BB}}}^{[:K]}} \right\|_F^2}\\
  {{\rm{s}}.{\rm{t}}.}&{{{{\bf{U}}^H_{{\rm{BB}}}}{{\bf{U}}_{{\rm{BB}}}} = {{\bf{U}}_{{\rm{BB}}}}{{\bf{U}}^H_{{\rm{BB}}}} = {{\bf{I}}_{{N_{{\rm{RF}}}}}}}}.
  \end{array}
  \end{equation}
With only the first $K$ columns of ${\bf{U}}_{{\rm{BB}}}$ considered in the objective function, the problem (37) is a typical semi-orthogonal procrustes problem \cite{OPP}, which can be solved by
  \begin{equation}
  {\bf{U}}_{{\rm{BB}}}^{[:K]} = {{\bf{U}}_{\rm{B}}}{{\bf{I}}_{{N_{{\rm{RF}}}}}^{[:K]}}{\bf{V}}_{\rm{B}}^H,
  \end{equation}
where ${{\bf{U}}_{\rm{B}}}{{\bf{\Sigma }}_{\rm{B}}}{\bf{V}}_{\rm{B}}^H = \widetilde{{\bf{F}}}_{{\rm{RF}}}^H \widetilde{{\bf{U}}}_{{\rm{opt}}}^{[:K]}$.

Generally, the AO method is initialized with a random optimization object. In (33), it could be $\widetilde{{\bf{F}}}_{{\rm{RF}}}$ with random phases or a random unitary ${\bf{U}}_{{\rm{BB}}}$. Since the optimization object gradually approaches the near-optimal in each iteration, the initialization scheme is closely related to the convergence time of the algorithm. Thus, we provide a simple scheme for the initialization of $\widetilde{{\bf{F}}}_{{\rm{RF}}}$. Regardless of the unitary constraint, the problem (33) can be decomposed into $N_{\rm{RF}}$ subproblems as
  \begin{equation}
  \begin{array}{*{20}{c}}
{\mathop {\min }\limits_{{\widetilde{{\bf{f}}}_n}} }&{\left\| {\sqrt M \widetilde{{\bf{U}}}_{{\rm{opt,}}n}^{[:K]} - {\widetilde{{\bf{f}}}_n}{\bf{U}}_{{\rm{BB}}}^{[:K]}(n,:)} \right\|_F^2}\\
{{\rm{s}}.{\rm{t}}.}&{\left| {{\widetilde{{\bf{f}}}_n}(m)} \right| = 1,}
\end{array}
  \end{equation}
where ${\widetilde{{\bf{U}}}_{{\rm{opt,}}n}^{[:K]}}\in {\mathbb{C}^{{M} \times K}}$ denotes the submatrix of $\widetilde{{\bf{U}}}_{{\rm{opt}}}^{[:K]} = {[{(\widetilde{{\bf{U}}}_{{\rm{opt,1}}}^{[:K]})^H},...,{(\widetilde{{\bf{U}}}_{{\rm{opt,}}{N_{{\rm{RF}}}}}^{[:K]})^H}]^H}$. The subproblem (39) is in the same form as (21), which suggests to initial $\widetilde{{\bf{F}}}_{{\rm{RF}}}$ with
  \begin{equation}
{\widetilde{{\bf{f}}}_{n,{\mathop{\rm int}} }}(m) = {e^{j\arg \{ {{\bf{U}}_n}(m,1)\} }},
  \end{equation}
where ${{\bf{U}}_n}{{\bf{\Sigma }}_n}{\bf{V}}_n^H = \widetilde{{\bf{U}}}_{{\rm{opt,}}n}^{[:K]}$.

As a summary, the AO-based scheme for updating clustering centers is shown in Algorithm 2. At the beginning of the scheme, $\widetilde{{\bf{F}}}_{{\rm{RF}}}$ is initialized by (40) to reduce the computational complexity, which will be further discussed in Section V. From step 4 to step 10, we alternately optimize $\widetilde{{\bf{F}}}_{{\rm{RF}}}$ and ${\bf{U}}_{{\rm{BB}}}$ until the difference between the objective function values in adjacent loops is little enough.

Since the objective function in (33) is successively minimized in step 6 and step 7, which is further lower bounded by zero, Algorithm 2 shall converge to a feasible solution. Although Algorithm 2 cannot guarantee the optimal solution due to the non-convexity of the problem, it can still provide a near optimal solution as shown in the simulation results in Section V. The computational complexity of algorithm 2 is mainly caused by the calculation of the SVD in step 5, where the dimension of the matrix is $N_{\rm{RF}} \times K$. Assuming that the number of iterations is $K_{\rm{iter}}$, the overall complexity can be given by $o(K_{\rm{iter}}(KN^2_{\rm{RF}}+K^3))$.

\emph{Remark}: The block diagonal ${\widetilde{{\bf{F}}}_{{\rm{RF}}}}$ implies that the AO-based method can be adopted to provide SHP design. Specifically, in Algorithm 2, it only needs to obtain ${\bf{U}}_{\rm{opt}}$ by SVD method at the beginning, and refine ${\bf{F}}_{\rm{BB}}$ by classical digital precoding method at the end. Essentially, SHP is a special case of AHP, where the fixed connection pattern indicates the relevance among antennas is not considered for RF chain design. Thus, with efficient RF chain design, it is common for AHP to provide advanced precoding performance than SHP, which will be further clarified in Section V.

\subsection{Proposed MKM-based AHP Scheme}

With the iteration of the two parts considered above, we propose the MKM-based AHP scheme in Algorithm 3. Initialized with a random ${\bf{U}}_{{\rm{BB}}}^{[:K]}$, the clustering is performed from step 4 to step 16. Due to the constraint (29), it is worth mentioning that greedy clustering each sample into the nearest cluster is not advisable, since it will lead to inconsistent numbers of members in different clusters. For the fairness among RF chains, which is similarly considered in \cite{AHPmu,MYpimrc}, we assign members to different clusters in turn from step 8 to step 11 as an inner loop. And the inner loop will be repeated $M$ times for the complete design. Then, the clustering centers are updated based on Algorithm 2 for further iteration, which will be terminated when the difference between the objective function values in adjacent loops reaches a certain threshold. At the end of the scheme, the baseband precoder is further refined by the classical digital precoding algorithms. Actually, similar to the proposed FHP scheme, the LS method is also operable for baseband precoder design in AHP. However, since the degree of freedom for RF precoder design significantly decreases in AHP, it is inadvisable to approximate the column space of ${\bf{F}}_{{\rm{RF}}}$ to that of ${\bf{F}}_{{\rm{opt}}}$ with LS method. Thus, the classical digital precoding is proposed.

In the proposed MKM-based AHP scheme, the optimal digital procoding is decomposed into $N_{\rm{t}}$ components, with each antenna corresponding to one of them. Based on the relevance among antennas, the similar components are merged into one cluster, which implies the similar antennas will be supported by the same RF chain.

\begin{algorithm}[!t]
\caption{Proposed MKM-based AHP algorithm}
\begin{algorithmic}[1]
\REQUIRE Full digital precoder ${\bf{F}}_{\rm{opt}}$, the number of RF chains $N_{\rm{RF}}$.
\ENSURE ${\bf{F}}_{\rm{RF}}$, ${\bf{F}}_{\rm{BB}}$.
\STATE Calculate the SVD ${\bf{F}}_{\rm{opt}}={{\bf{U}}_{{\rm{opt}}}}{{\bf{\Sigma }}_{{\rm{opt}}}}{\bf{V}}_{{\rm{opt}}}^H$.
\STATE Initialize ${\bf{U}}_{{\rm{BB}}}^{[:K]}$ as a random semi-unitary matrix.
\STATE Initialize the value of objective function as infinite $v_{0}=+\infty$, and set $k=1$.
\REPEAT
\STATE For each pair of $\sqrt M{\bf{U}}_{{\rm{opt}}}^{[:K]}(m,:)$ and ${\bf{U}}_{{\rm{BB}}}^{[:K]}(n,:)$, calculate the distance based on (32).
\STATE ${\bf{F}}_{\rm{RF}}={\bf{0}}_{{N_{\rm{t}}}\times {N_{\rm{RF}}}}$, $\Gamma  = \{ 1,2,...,{N_{\rm{t}}}\} $, ${\Gamma _n} = \emptyset ,\forall n$.
\FOR {$p = 1$ to $M$}
\FOR {$q = 1$ to ${N_{\rm{RF}}}$}
\STATE ${i^ \star }\mathop { = \arg \min }\limits_{i \in \Gamma } {{\cal D}_{\rm{A}}}(\sqrt M {\bf{U}}_{{\rm{opt}}}^{[:K]}(i,:),{\bf{U}}_{{\rm{BB}}}^{[:K]}(q,:))$.
\STATE ${\Gamma _q} = {\Gamma _q}  \cup  \{ {i^ \star }\} $, $\Gamma  = \Gamma \backslash \{ {i^ \star }\} $.
\ENDFOR
\ENDFOR
\STATE Fixed ${\Gamma _n}$, refine ${\bf{U}}_{{\rm{opt}}}^{[:K]}$ and ${\bf{F}}_{\rm{RF}}$ by Algorithm 2.
\STATE $k=k+1$.
\STATE $v_{k}= {\left\| {\sqrt M {\bf{U}}_{{\rm{opt}}}^{[:K]} - {{\bf{F}}_{{\rm{RF}}}}{\bf{U}}_{{\rm{BB}}}^{[:K]}} \right\|_F^2} $.
\UNTIL ${v_{k - 1}} - {v_k} < \varepsilon $
\STATE Obtain ${\bf{F}}_{\rm{BB}}$ by the classical digital precoding method (15) with the effective baseband channel.
\end{algorithmic}
\end{algorithm}

The convergence of Algorithm 3 can be similarly clarified as Algorithm 2, since it refines the objective function value in the iterative refinement. The computational complexity of Algorithm 3 is mainly caused by the repetition of the Algorithm 2. With $K_{\rm{out}}$ outer iterations\footnote{The average number of outer iterations for convergence in numerical simulation of Section V is approximately $K_{\rm{out}}=6$.}, the overall complexity can be similarly given by $o(K_{\rm{out}}K_{\rm{iter}}(KN^2_{\rm{RF}}+K^3))$.


\emph{Remark}: The HAC method is not advisable to solve (28). Since the key approach of HAC is merging clusters until the desired number of clusters is obtained, it is of high possibility to merge two large clusters in the last few loops, which makes the clustering result improbable to satisfy the constraint (29).

\section{Simulation}
In this section, we present the numerical results based on Monte Carlo simulations to evaluate the performance of proposed HAC-based FHP and MKM-based AHP schemes. In addition, the proposed SHP design is provided by the AO-based algorithm in Section IV-B for more details of performance comparison. In the simulation system, the BS is equipped with $8 \times 8$ antenna array to serve $8$ users, while each user is equipped with $2 \times 2$ antenna array. The antenna elements are separated by $d=\lambda / 2$ in UPA structure. The mmWave channel between the BS and each user consists of $N_{\rm{c},k}=5$ scattering clusters, each of which contains $N_{\rm{p},k}=10$ propagation paths. The azimuth and elevation angles of arrival and departure are uniformly distributed in $[0,2\pi)$ with a 10-degree angular spread. All simulation results are calculated over 1000 channel realizations.

\subsection{System Spectral Efficiency}
In this subsection, we investigate the system spectral efficiency of different hybrid structures achieved by proposed design schemes. The ZF method is utilized to provide optimal full digital precoder as input, and the baseband precoder revision scheme for the last step of each proposed method. In addition, the unquantized APSs are adopted at both combiners and precoders. For further comparison, the FHP scheme in \cite{fuldecp1}, the AHP scheme in \cite{AHPmu} and the SHP scheme in \cite{subdecp} are also presented as benchmarks. Here, to enable the existing schemes, we consider the case that the number of RF chains is equal to that of the users, i.e., $N_{\rm{RF}}=K$. As shown in Fig. 2, the fully-connected structure always provides the best system performance. The proposed HAC-based scheme shows a slight advantage compared with the scheme in \cite{fuldecp1}. As for AHP, the proposed MKM-based scheme achieves much higher spectral efficiency than the scheme in \cite{AHPmu}. The increasing SNR witnesses dramatic expansion of the performance gap. Especially when $\rm{SNR}>0$ $\rm{dB}$, the spectral efficiency gap is over $10$ $\rm{bps/Hz}$. For the sub-connected structure, similar conclusions can be obtained in the comparison of the proposed AO-based scheme and the scheme in \cite{subdecp}. Normally, the adaptively-connected structure shall provide better precoding performance than the sub-connected structure with flexible connection among antennas and RF chains. However, the proposed AO-based SHP scheme even outperforms the AHP scheme in \cite{AHPmu}. This is mainly because the proposed scheme jointly optimizes the baseband and RF precoder, while the existing scheme decouples the optimization problem.

\begin{figure}[!t]
\centering
\includegraphics[width=8cm]{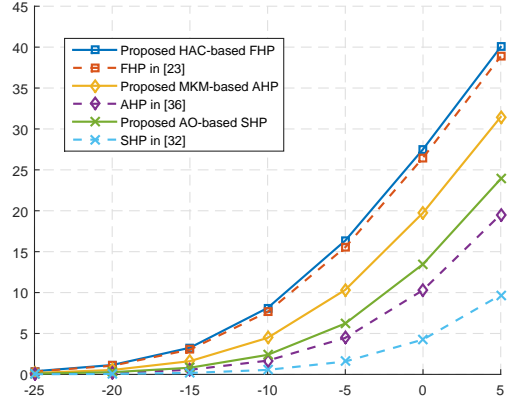}
\caption{Spectral efficiency versus SNR with different precoding schemes when $N_{\rm{t}}=8 \times 8$, $K=N_{\rm{RF}}=8$.} \label{fig.2}
\end{figure}

\begin{figure}[!t]
\centering
\includegraphics[width=8cm]{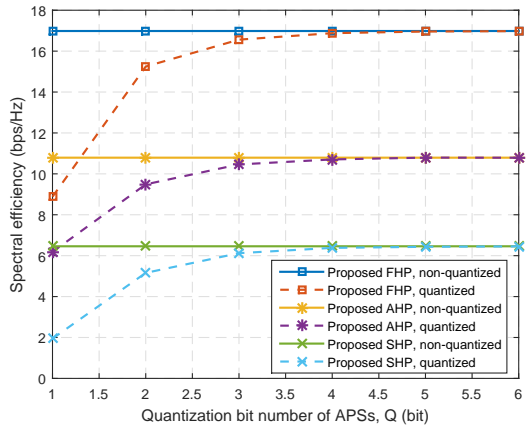}
\caption{Spectral efficiency versus the quantization bit number $Q$ with different proposed precoding schemes when $\rm{SNR}=-5~\rm{dB}$.} \label{fig.3}
\end{figure}

Fig. 3 shows the impact of the quantized APSs on the system spectral efficiency. All three hybrid structures designed by the proposed schemes are considered for the comparison when $\rm{SNR}=-5~\rm{dB}$. By increasing the quantization bits of APSs, the performance of quantized schemes rapidly approaches that of non-quantized schemes owing to the improvement of phase resolution in RF precoder. When $Q=4$, the curves of quantized schemes almost coincide with that of corresponding non-quantized schemes with the gap at about $0.1~\rm{bps/Hz}$, which provides guidance for practical APS design. Moreover, since AHP and SHP share the same number of nonzero entries in the RF precoding matrix ($N_{\rm{t}}$ in total), the impacts of few quantization bits are at almost the same level. However, the fully-connected structure with few quantization bits suffers a more serious performance loss, since all entries in RF precoding matrix are non-zero ($N_{\rm{t}}N_{\rm{RF}}$ in total).

Fig. 4 shows the impact of the array scale at the BS with the comparison $N_{\rm{t}}=\{8\times8,8\times32\}$. The proposed MKM-based AHP scheme is adopted to provide the precoder design. Obviously, growing spectral efficiency can be observed with the enlargement of the antenna array. In addition, Fig. 4 shows the impact of target full digital precoder provided by different digital precoding methods including MF, ZF and RZF. For consistency, at the end of the algorithm, the baseband precoder is revised by the same method as getting the input. Compared with the ZF scheme, the RZF scheme shows slight advantage in low SNR conditions, owing to the consideration of the noise parameter. The increasing SNR weakens the influence of the noise parameter in (14) and makes the performance gap disappear. The MF scheme provides the poorest performance at high SNR conditions, since it only harvests the channel gain but is incapable to eliminate the inter-user interference.

\begin{figure}[!t]
\centering
\includegraphics[width=8cm]{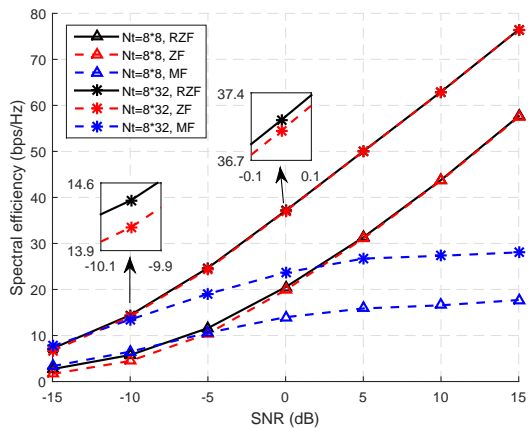}
\caption{Spectral efficiency of proposed MKM-based AHP scheme with different target full digital precoder when $N_{\rm{t}}=\{8\times8,8\times32\}$.} \label{fig.4}
\end{figure}

To clarify the performance of the proposed HAC-based FHP scheme in the cases of insufficient RF chains, in Fig. 5, we present the spectral efficiency versus the number of RF chains when $\rm{SNR}=-10~\rm{dB}$. As proposed in Section III, we adopt the LS and ZF schemes to refine the baseband precoder and analyze the performance difference, respectively. The FHP scheme in \cite{fuldecp1} and the upper bounds in (23) are also presented as the benchmarks. It is observed that the spectral efficiencies of the proposed scheme with the LS and ZF refinement schemes gradually get close to the corresponding upper bounds, which further verifies the tightness especially in the cases of large $N_{\rm{RF}}$. However, the spectral efficiency of the scheme in \cite{fuldecp1} maintains at a constant level. The defect results from the limitation that only $K$ RF chains are effectively used, while other $N_{\rm{RF}}-K$ RF chains keep silence and make no difference to the system performance.

\begin{figure}[!t]
\centering
\includegraphics[width=8cm]{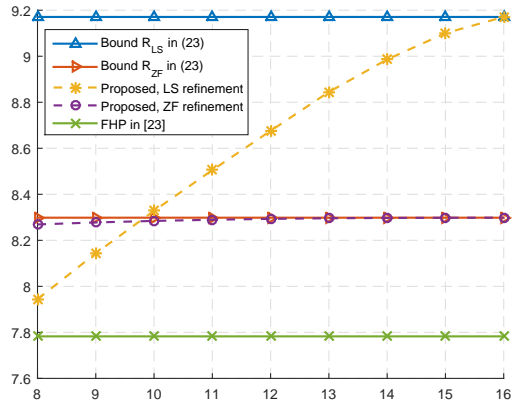}
\caption{Spectral efficiency of proposed HAC-based FHP scheme versus the number of RF chains with different baseband precoder refinement schemes when $\rm{SNR}=-10~\rm{dB}$.} \label{fig.5}
\end{figure}

A more important conclusion from Fig. 5 is that the spectral efficiency of the LS scheme suffers a much sharper decline with the decreasing $N_{\rm{RF}}$ than the ZF scheme. This is because the insufficient RF chains lead to low similarity of column spaces between the proposed RF precoder and the full digital precoder, which reduces the approximation accuracy of the LS method and exacerbates the inter-user interference. As for the ZF scheme, the inter-user interference can always be eliminated, which abates the influence of column space similarity to a certain extent. Further, a threshold of $N_{\rm{RF}}$ for the selection of two alternative refinement schemes can be given based on the intersection point of two performance curves. In the given simulation, the ZF method shall be adopted when $N_{\rm{RF}} \leq 10$, otherwise it shall be the LS method.

\subsection{System Power Efficiency}
In the previous simulation, the FHP shows the best performance in terms of system spectral efficiency. However, the fully-connected structure also leads to the highest power consumption. In this subsection, the power consumption is further taken into account for the comparison of different precoding schemes.

The power consumption of FHP can be given by \cite{AMmethod}
\begin{equation}
{P_{{\rm{ful}}}} = {P_{{\rm{com}}}} + {N_{{\rm{RF}}}}{P_{{\rm{RF}}}} + {N_{\rm{t}}}{P_{{\rm{PA}}}} + {N_{{\rm{RF}}}}{N_{\rm{t}}}{P_{{\rm{APS}}}},
\end{equation}
where ${P_{{\rm{com}}}}$ is the common power of the transmitter, ${P_{{\rm{RF}}}}$, ${P_{{\rm{PA}}}}$, and ${P_{{\rm{APS}}}}$ are the power of a single RF chain, power amplifier, and APS, respectively. For SHP, which only requires ${N_{{\rm{t}}}}$ APSs, the power consumption can be given by
\begin{equation}
{P_{{\rm{sub}}}} = {P_{{\rm{com}}}} + {N_{{\rm{RF}}}}{P_{{\rm{RF}}}} + {N_{\rm{t}}}{P_{{\rm{PA}}}} + {N_{\rm{t}}}{P_{{\rm{APS}}}}.
\end{equation}
As for AHP, it further requires ${N_{{\rm{t}}}}$ switches to implement the adaptive connection network, and incurs power consumption given by
\begin{equation}
{P_{{\rm{adp}}}} = {P_{{\rm{com}}}} + {N_{{\rm{RF}}}}{P_{{\rm{RF}}}} + {N_{\rm{t}}}{P_{{\rm{PA}}}} + {N_{\rm{t}}}({P_{{\rm{APS}}}}+{P_{{\rm{SW}}}}),
\end{equation}
where ${P_{{\rm{SW}}}}$ is the power of a single switch. The power efficiency of a certain precoding design can be defined as the ratio between the spectral efficiency and power consumption, i.e., $\eta  = R/P$.

\begin{figure}[!t]
\centering
\includegraphics[width=8cm]{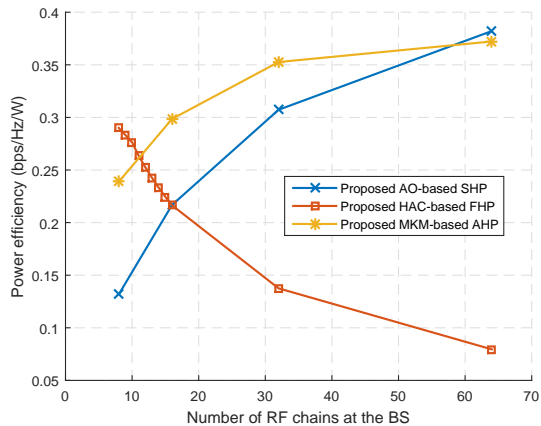}
\caption{Power efficiency versus the number of RF chains with different proposed precoding schemes when $\rm{SNR}=-10~\rm{dB}$.} \label{fig.6}
\end{figure}

In Fig. 6, we present the power efficiency versus the number of RF chains for three hybrid precoding structures with the power parameters set as ${P_{{\rm{com}}}}=10~\rm{W}$, ${P_{{\rm{RF}}}}=100~\rm{mW}$, ${P_{{\rm{PA}}}}=100~\rm{mW}$, ${P_{{\rm{APS}}}}=20~\rm{mW}$, ${P_{{\rm{SW}}}}=10~\rm{mW}$ according to \cite{AMmethod,Power}. For the proposed FHP scheme with optimal baseband precoder refinement, although the spectral efficiency sufficiently approaches that of the optimal full digital precoder with ${N_{{\rm{RF}}}}$ increasing from 8 to 16 in Fig. 5, the performance growth rate does not exceed $20\%$. However, the power consumption increases by $40\%$, which results in a negative impact on the power efficiency. Thus, a downward trend can be observed in Fig. 6. If the RF chains are over equipped, i.e., ${N_{{\rm{RF}}}} > 2K$, the spectral efficiency will be limited by the performance of full digital precoder since the extra RF chains cannot be efficiently utilized. Consequently, the FHP design suffers a dramatic decrease of the power efficiency with the increasing RF chains.

For both the proposed AHP and SHP schemes, in Fig. 6, the curves of power efficiency monotonically increase with ${N_{{\rm{RF}}}}$. The increasing spectral efficiency is one of the reasons for the improvement of power efficiency. More importantly, the almost unchanged power consumption gives rise to the positive impact, since the increasing RF chains will not result in any further requirements of APSs. In the cases of ${N_{{\rm{RF}}}} < {N_{{\rm{t}}}}$, the AHP scheme achieves much higher power efficiency than the SHP scheme. This is mainly because the adaptive connection network significantly increases the degree of freedom for the RF precoder design with several switchers which require low power consumption. At the special point ${N_{{\rm{RF}}}} = {N_{{\rm{t}}}}$ (extremely impractical in massive MIMO systems), the SHP scheme slightly outperforms the AHP scheme. This results from that the RF precoders in AHP and SHP are in the same form when ${N_{{\rm{RF}}}} = {N_{{\rm{t}}}}$, which degenerates the adaptively-connected structure into the sub-connected structure and makes the extra switches unable to provide any flexibility.

\subsection{Initialization Scheme Analysis for AO-based Algorithm}
As mentioned in Section IV-B, we provide a specific initialization scheme for proposed AO-based algorithm. In this subsection, we will show the impact of the initialization scheme in terms of the spectral efficiency and computational complexity. Both AHP and SHP schemes are considered for the comparison.

\begin{figure}[!t]
\captionsetup{belowskip=-10pt}
    \begin{minipage}[t]{0.5\linewidth}
    \centering
    \includegraphics[width=3.8cm]{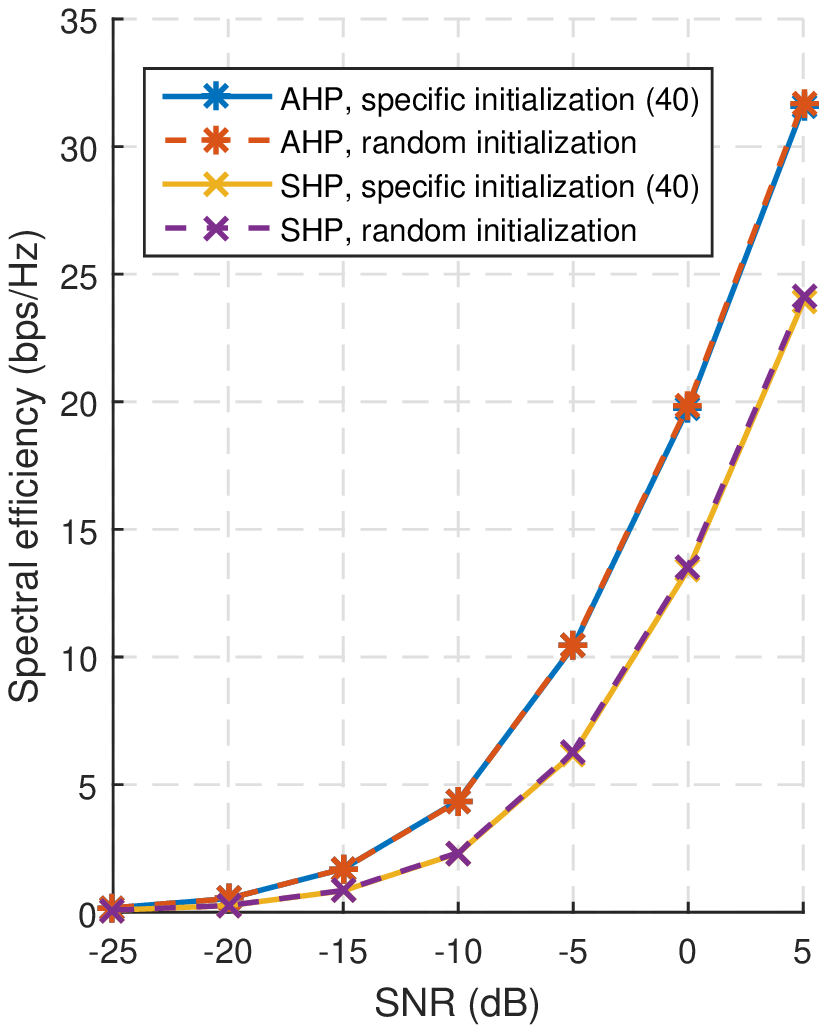}
    \scriptsize {\centerline{(a)}}
    \end{minipage}%
    \begin{minipage}[t]{0.5\linewidth}
    \centering
    \includegraphics[width=3.8cm]{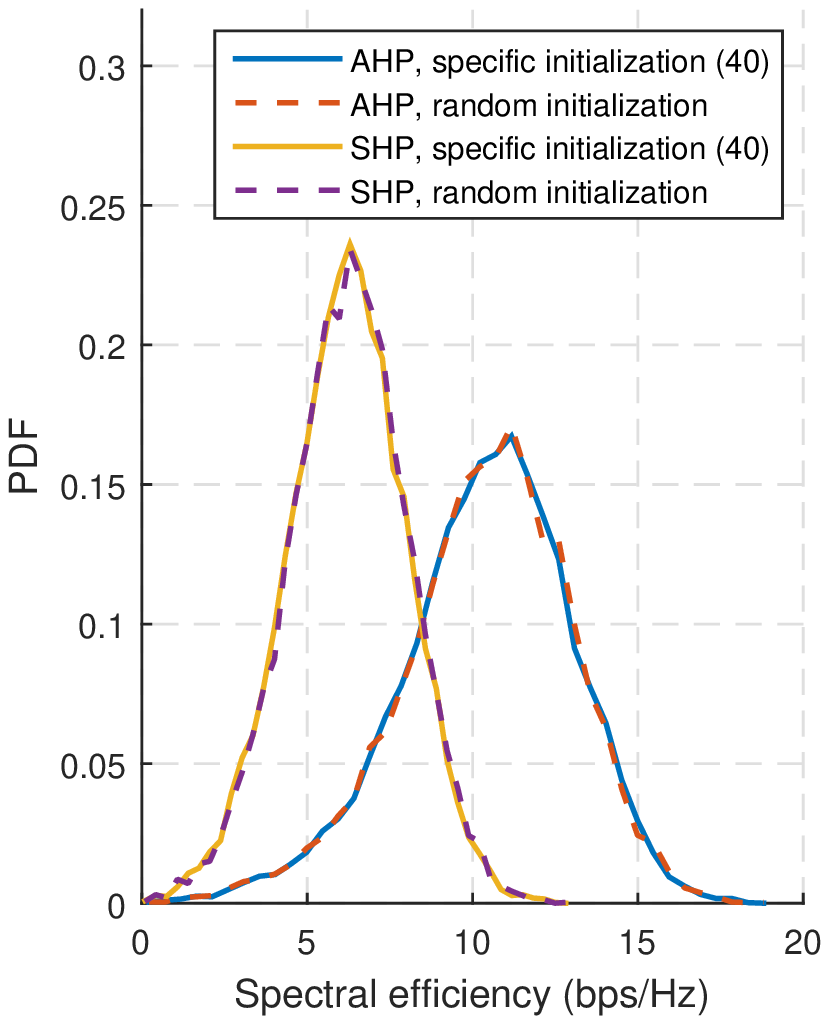}
    \scriptsize {\centerline{(b)}}
    \end{minipage}%
    \caption{(a) Spectral efficiency of proposed SHP and AHP schemes with different initialization schemes; (b) The PDF of spectral efficiency when $\rm{SNR}=-5~\rm{dB}$.}
    \label{fig.7}
 \end{figure}

Firstly, we compare the proposed AHP and SHP schemes with random and specific initialization schemes in terms of the spectral efficiency. The random initialization scheme means the phases in APSs are randomly initialized, which is generally adopted in classical AO algorithm. The specific initialization scheme refers to the proposed initialization scheme in (40). Fig. 7(a) clarifies a common conclusion for AHP and SHP that the different initialization schemes contribute to almost the same average spectral efficiency. For more details, Fig. 7(b) plots the probability density function (PDF) curves of the spectral efficiency for the specific point $\rm{SNR}=-5~\rm{dB}$ in Fig. 7(a). The random and specific initialization schemes provide almost the same regularity of distribution.


\begin{figure}[!t]
\captionsetup{belowskip=-10pt}
    \begin{minipage}[t]{0.5\linewidth}
    \centering
    \includegraphics[width=3.8cm]{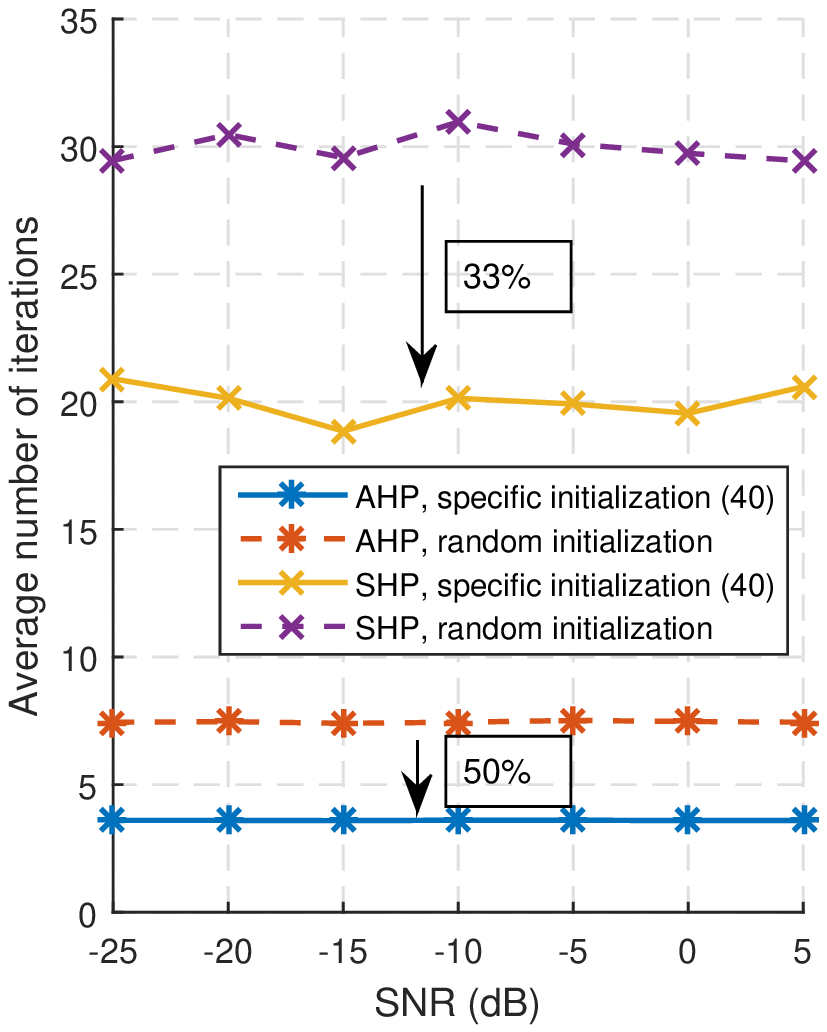}
    \scriptsize {\centerline{(a)}}
    \end{minipage}%
    \begin{minipage}[t]{0.5\linewidth}
    \centering
    \includegraphics[width=3.8cm]{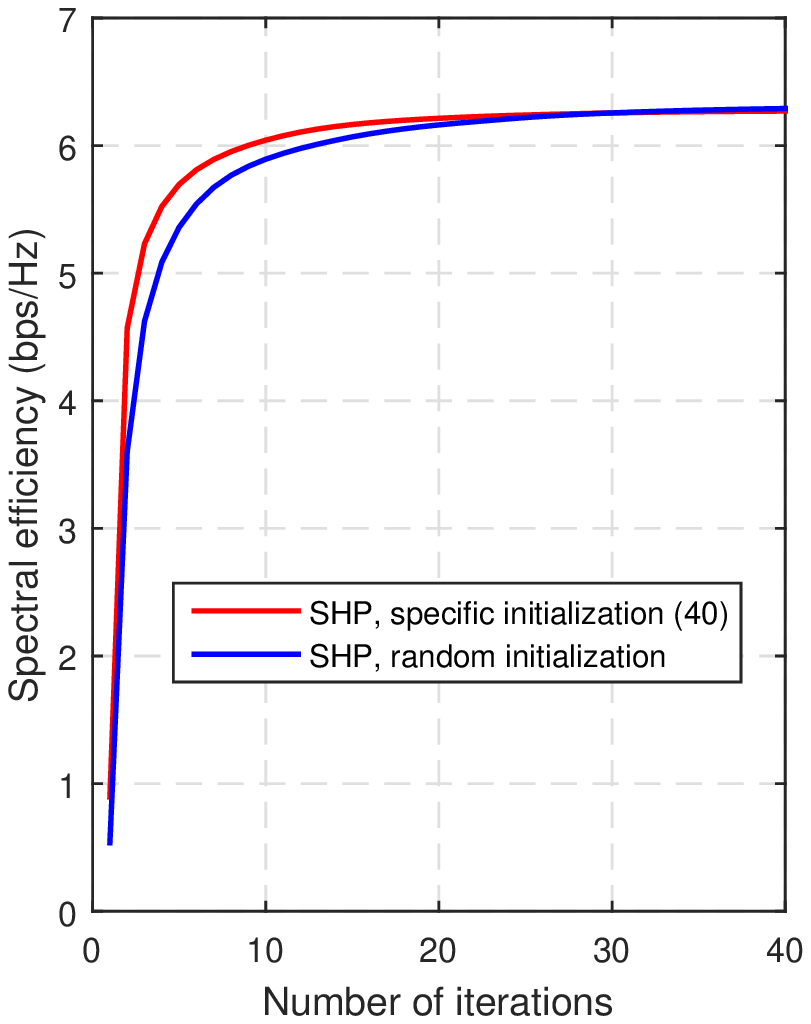}
    \scriptsize {\centerline{(b)}}
    \end{minipage}%
    \caption{(a) Average number of iterations of proposed SHP and AHP schemes versus SNR; (b) The spectral efficiency of the proposed SHP scheme versus the number of iterations when $\rm{SNR}=-5~\rm{dB}$.}
    \label{fig.8}
 \end{figure}


To compare the computational complexity of the algorithms with two alternative initialization schemes, in Fig. 8(a), we count the average number of iterations $K_{\rm{iter}}$ in the AO-based algorithm. For SHP, the proposed specific initialization scheme requires $33\%$ less iterations for convergence than the random initialization scheme. For AHP, the rate of reduction even reaches $50\%$. Further in Fig. 8(b), we present the spectral efficiency of the AO-based SHP scheme versus the number of iterations. It demonstrates that the specific initialization scheme contributes to less convergence time than the random initialization scheme.

\section{Conclusion}
In this paper, we proposed hybrid precoding schemes for fully-connected and adaptively-connected structures in multi-user massive MIMO systems. Innovatively, the thought of clustering in unsupervised learning was adopted. Based on the relevance among RF chains in optimal hybrid procoder, the HAC-based FHP scheme was proposed to provide the design with insufficient RF chains. Based on the relevance among antennas, the MKM-based AHP scheme was proposed to make full use of the flexibility in the structure. Particularly, the AO-based clustering center updating algorithm with a specific initialization scheme in the MKM-based scheme was capable to provide feasible SHP design individually, since SHP can be regarded as a special case of AHP.

Simulation results have illustrated that the HAC-based FHP scheme can achieve close spectral efficiency to the full digital precoder with insufficient RF chains. The MKM-based AHP scheme provides high power efficiency with the full use of the low cost adaptive connection network. With the lowest power consumption, the AO-based SHP scheme effectively utilizes the equipped RF chains to enhance both spectral and energy efficiency. Moreover, all proposed precoding schemes outperform the existing work for corresponding hybrid structure designs. The specific initialization scheme for the AO-based algorithm dramatically reduces the convergence time of AHP and SHP schemes.

Finally, our work has proved the feasibility of clustering method in hybrid precoding design. However, only two representative clustering algorithms (HAC and K-means) are considered in this paper. Actually, there is a great possibility for other clustering algorithms to achieve better performance, which will require further investigation.


\appendices
\section{Proof of the Proposition 1}
With the adjusted RF precoding matrix denoted by $\bar{{\bf{F}}}_{\rm{RF}} \in \mathbb{C}^{N_{{\rm{t}}} \times {\bar{N}}_{{\rm{RF}}}}$, the unchanged column space implies that the columns in $\bar{{\bf{F}}}_{\rm{RF}}$ can be linear represented by the columns in ${{\bf{F}}}_{\rm{RF}}$, i.e., $\bar{{\bf{F}}}_{\rm{RF}}={{\bf{F}}}_{\rm{RF}}{{\bf{A}}}$, where ${{\bf{A}}} \in \mathbb{C}^{N_{{\rm{RF}}} \times {\bar{N}}_{{\rm{RF}}}}$ is an arbitrary matrix.

Further, the basis vectors of the column space of ${{\bf{F}}}_{\rm{RF}}$ and $\bar{{\bf{F}}}_{\rm{RF}}$ can be respectively given by the columns of ${{\bf{U}}^{[:r]}_{\rm{RF}}}$ and ${\bar{\bf{U}}}^{[:r]}_{\rm{RF}}$, where ${{\bf{U}_{\rm{RF}}}{\bf{\Sigma}_{\rm{RF}}}{\bf{V}}^H_{\rm{RF}}}={{\bf{F}}}_{\rm{RF}}$, ${{\bf{\bar U}}_{{\rm{RF}}}}{{\bf{\bar \Sigma }}_{{\rm{RF}}}}{\bf{\bar V}}_{{\rm{RF}}}^H = {{\bf{\bar F}}_{{\rm{RF}}}}$, and $r={{\rm{r}}({{\bf{F}}}_{\rm{RF}})} = {{\rm{r}}({\bar{\bf{F}}}_{\rm{RF}})}$. And the standard bases can be represented by each other with unitary transformation, i.e., ${\bar{\bf{U}}}^{[:r]}_{\rm{RF}}={{\bf{U}}}^{[:r]}_{\rm{RF}}{\bf{\Psi}} $, where ${\bf{\Psi}} \in \mathbb{C}^{r \times r}$ is a unitary matrix.

Accordingly, we have the following derivation
\begin{equation}
  \begin{split}
{{\bf{F}}_{{\rm{RF}}}}{\bf{A}} &= {{{\bf{\bar F}}}_{{\rm{RF}}}}\\
{\bf{U}}_{{\rm{RF}}}^{[:r]}{\bf{\Sigma }}_{{\rm{RF}}}^{[r:]}{\bf{V}}_{{\rm{RF}}}^H{\bf{A}} &= {\bf{\bar U}}_{{\rm{RF}}}^{[:r]}{\bf{\bar \Sigma }}_{{\rm{RF}}}^{[r:]}{\bf{\bar V}}_{{\rm{RF}}}^H\\
{\bf{U}}_{{\rm{RF}}}^{[:r]}{\bf{\Sigma }}_{{\rm{RF}}}^{[r:]}{\bf{V}}_{{\rm{RF}}}^H{\bf{A}} &= {\bf{U}}_{{\rm{RF}}}^{[:r]}{\bf{\Psi \bar \Sigma }}_{{\rm{RF}}}^{[r:]}{\bf{\bar V}}_{{\rm{RF}}}^H\\
{\bf{\Sigma }}_{{\rm{RF}}}^{[r:]}{\bf{V}}_{{\rm{RF}}}^H{\bf{A}} &= {\bf{\Psi \bar \Sigma }}_{{\rm{RF}}}^{[r:]}{\bf{\bar V}}_{{\rm{RF}}}^H.
  \end{split}
\end{equation}
Note that ${\bf{P}} = {\bf{\Sigma }}_{{\rm{RF}}}^{[r:]}{\bf{V}}_{{\rm{RF}}}^H{\bf{A}}$ has full row rank and satisfies ${\bf{P}}{{\bf{P}}^\dag } = {\bf{I}}$, we adjust the baseband precoding matrix as ${{\bf{\bar F}}_{{\rm{BB}}}} = {{\bf{P}}^\dag }{\bf{\Sigma }}_{{\rm{RF}}}^{[r:]}{\bf{V}}_{{\rm{RF}}}^H{{\bf{F}}_{{\rm{BB}}}}$. Thus, the adjusted hybrid precoding matrix satisfies
\begin{equation}
  \begin{split}
{{{\bf{\bar F}}}_{{\rm{RF}}}}{{{\bf{\bar F}}}_{{\rm{BB}}}} &= {{\bf{F}}_{{\rm{RF}}}}{\bf{A}}{{\bf{P}}^\dag }{\bf{\Sigma }}_{{\rm{RF}}}^{[r:]}{\bf{V}}_{{\rm{RF}}}^H{{\bf{F}}_{{\rm{BB}}}}\\
 &= {\bf{U}}_{{\rm{RF}}}^{[:r]}{\bf{\Sigma }}_{{\rm{RF}}}^{[r:]}{{\bf{V}}_{{\rm{RF}}}^H}{\bf{A}}{{\bf{P}}^\dag }{\bf{\Sigma }}_{{\rm{RF}}}^{[r:]}{\bf{V}}_{{\rm{RF}}}^H{{\bf{F}}_{{\rm{BB}}}}\\
 &= {\bf{U}}_{{\rm{RF}}}^{[:r]}{\bf{P}}{{\bf{P}}^\dag }{\bf{\Sigma }}_{{\rm{RF}}}^{[r:]}{\bf{V}}_{{\rm{RF}}}^H{{\bf{F}}_{{\rm{BB}}}}\\
 &= {{\bf{F}}_{{\rm{RF}}}}{{\bf{F}}_{{\rm{BB}}}},
  \end{split}
\end{equation}
which means no performance difference.

\section{Proof of the Theorem 1}
According to Proposition 1, the ideal case for the HAC-based FHP scheme is that ${{\bf{F}}_{{\rm{RF}}}}$ holds the same column space as ${{\bf{F}}_{{\rm{opt}}}}$, i.e., ${\bf{U}}_{{\rm{RF}}}^{[:r]} = {\bf{U}}_{{\rm{opt}}}^{[:r]}{\bf{\Theta }}$, where ${\bf{U}}_{{\rm{RF}}}$ and ${\bf{U}}_{{\rm{opt}}}$ denote the left-singular matrices of ${{\bf{F}}_{{\rm{RF}}}}$ and ${{\bf{F}}_{{\rm{opt}}}}$, $r = {\rm{r(}}{{\bf{F}}_{{\rm{RF}}}}{\rm{)}} = {\rm{r(}}{{\bf{F}}_{{\rm{opt}}}}) = {\rm{r(}}{{\bf{H}}_{{\rm{eq}}}})$, ${\bf{\Theta}} \in \mathbb{C}^{r \times r}$ is a unitary matrix.

With ${{\bf{F}}_{{\rm{BB}}}} = {\bf{F}}_{{\rm{RF}}}^\dag {{\bf{F}}_{{\rm{opt}}}}$ obtained by the LS refinement scheme, we have the following derivation
\begin{equation}
\begin{split}
{{\bf{F}}_{{\rm{RF}}}}{{\bf{F}}_{{\rm{BB}}}} &= {{\bf{F}}_{{\rm{RF}}}}{\bf{F}}_{{\rm{RF}}}^\dag {{\bf{F}}_{{\rm{opt}}}}\\
 &= {\bf{U}}_{{\rm{RF}}}^{[:r]}{({\bf{U}}_{{\rm{RF}}}^{[:r]})^H}{\bf{U}}_{{\rm{opt}}}^{[:r]}{\bf{\Sigma }}_{{\rm{opt}}}^{[r:]}{\bf{V}}_{{\rm{opt}}}^H\\
 &= {\bf{U}}_{{\rm{opt}}}^{[:r]}{\bf{\Theta}}{({\bf{U}}_{{\rm{opt}}}^{[:r]}{\bf{\Theta}})^H}{\bf{U}}_{{\rm{opt}}}^{[:r]}{\bf{\Sigma }}_{{\rm{opt}}}^{[r:]}{\bf{V}}_{{\rm{opt}}}^H\\
 &= {{\bf{F}}_{{\rm{opt}}}} = \sqrt {K/\left\| {{{\bf{H}}_{{\rm{eq}}}}} \right\|_F^2} {\bf{H}}_{{\rm{eq}}}^\dag.
\end{split}
\end{equation}
Hence, the overall received signal in (3) can be given by
\begin{equation}
{\bf{r}} = \sqrt {K/\left\| {{\bf{H}}_{{\rm{eq}}}^\dag } \right\|_F^2} {{\bf{H}}_{{\rm{eq}}}}{\bf{H}}_{{\rm{eq}}}^\dag {\bf{x}} + {\bf{n}},
\end{equation}
where ${\bf{r}} = {[{r_1},...,{r_K}]^T}$, ${\bf{n}} = {[{\bf{w}}_1^H{{\bf{n}}_1},...,{\bf{w}}_K^H{{\bf{n}}_K}]^T}$. Since the inter-user interference is completely eliminated in (47), the $\rm{SINR}$ of the $k$-th user is ${\rm{SIN}}{{\rm{R}}_k} = P/({\sigma ^2}\left\| {{\bf{H}}_{{\rm{eq}}}^\dag } \right\|_F^2)$. Thus, the system spectral efficiency is bounded by
\begin{equation}
R_{\rm{LS}} \le K{\log _2}(1 + \frac{P}{{{\sigma ^2}\left\| {{\bf{H}}_{{\rm{eq}}}^\dag } \right\|_F^2}}).
\end{equation}

With ${{\bf{F}}_{{\rm{BB}}}} = {\bf{H}}_{{\rm{BB}}}^\dag  = {({{\bf{H}}_{{\rm{eq}}}}{{\bf{F}}_{{\rm{RF}}}})^\dag }$ obtained by the ZF refinement scheme, the received signal can be given by
\begin{equation}
{\bf{r}} = \sqrt {K/\left\| {{{\bf{F}}_{{\rm{RF}}}}{\bf{H}}_{{\rm{BB}}}^\dag } \right\|_F^2} {{\bf{H}}_{{\rm{eq}}}}{{\bf{F}}_{{\rm{RF}}}}{\bf{H}}_{{\rm{BB}}}^\dag {\bf{x}} + {\bf{n}},
\end{equation}
which also eliminates the inter-user interference and further bounds the system spectral efficiency as follows
\begin{equation}
R_{\rm{ZF}} \le K{\log _2}(1 + \frac{P}{{{\sigma ^2}\left\| {{\bf{F}}_{{\rm{RF}}} }{{\bf{H}}_{{\rm{BB}}}^\dag } \right\|_F^2}}).
\end{equation}

Since the row space of ${\bf{H}}_{\rm{eq}}$ is consistent with the column space of ${\bf{F}}_{\rm{opt}}$, the standard basis vectors of the row space of ${\bf{H}}_{\rm{eq}}$ can be linear represented by the columns of ${\bf{F}}_{\rm{RF}}$ with ${\bf{T}} \in \mathbb{C}^{{N_{\rm{RF}}} \times r}$, i.e.,
\begin{equation}
\begin{array}{c}
{{\bf{F}}_{{\rm{RF}}}}{\bf{T}} = {{\bf{V}}_{\rm{eq}}},\\
{\bf{V}}_{\rm{eq}}^H{{\bf{F}}_{{\rm{RF}}}}{\bf{T}} = {{\bf{I}}_r},
\end{array}
\end{equation}
where ${{\bf{U}}_{\rm{eq}}}{{\bf{\Sigma }}_{\rm{eq}}}{\bf{V}}_{\rm{eq}}^H = {{\bf{H}}_{{\rm{eq}}}}$, and ${{\bf{\Sigma }}_{\rm{eq}}} \in \mathbb{C}^{r \times r}$ is a diagonal matrix with $r$ nonzero singular values of ${{\bf{H}}_{{\rm{eq}}}}$ on the diagonal. According to (51), since ${N_{\rm{RF}}} \ge r$, the matrix ${\bf{V}}_{\rm{eq}}^H{{\bf{F}}_{{\rm{RF}}}}$ has full row rank and satisfies ${\bf{V}}_{\rm{eq}}^H{{\bf{F}}_{{\rm{RF}}}}{({\bf{V}}_{\rm{eq}}^H{{\bf{F}}_{{\rm{RF}}}})^\dag } = {\bf{I}}$.

Further, we derive as follows
\begin{equation}
  \begin{split}
  \left\| {{\bf{H}}_{{\rm{eq}}}^\dag } \right\|_F^2 &= \left\| {{\bf{\Sigma }}_{\rm{eq}}^{ - 1}{\bf{U}}_{\rm{eq}}^H} \right\|_F^2\\
&= \left\| {{\bf{V}}_{\rm{eq}}^H{{\bf{F}}_{{\rm{RF}}}}{{({\bf{V}}_{\rm{eq}}^H{{\bf{F}}_{{\rm{RF}}}})}^\dag }{\bf{\Sigma }}_{\rm{eq}}^{ - 1}{\bf{U}}_{\rm{eq}}^H} \right\|_F^2\\
&\mathop  \le \limits^{(a)} \left\| {{{\bf{F}}_{{\rm{RF}}}}{{({\bf{V}}_{\rm{eq}}^H{{\bf{F}}_{{\rm{RF}}}})}^\dag }{\bf{\Sigma }}_{\rm{eq}}^{ - 1}{\bf{U}}_{\rm{eq}}^H} \right\|_F^2\\
&\mathop  = \limits^{(b)} \left\| {{{\bf{F}}_{{\rm{RF}}}}{{({{\bf{H}}_{{\rm{eq}}}}{{\bf{F}}_{{\rm{RF}}}})}^\dag }} \right\|_F^2 = \left\| {{{\bf{F}}_{{\rm{RF}}}}{{\bf{H}}^\dag_{{\rm{BB}}}} } \right\|_F^2,
  \end{split}
\end{equation}
where (a) comes to equality if and only if ${{\bf{V}}_{\rm{eq}}}$ fits the left-singular matrix of ${{\bf{F}}_{\rm{RF}}}{{\bf{H}}^\dag_{{\rm{BB}}}}$, and (b) results from the product property of pseudo inverse. According to (52), $R_{\rm{ZF}} \leq R_{\rm{LS}}$ is always satisfied.

\ifCLASSOPTIONcaptionsoff
  \newpage
\fi

{
\small
\bibliographystyle{ieeetr}
\bibliography{reference}

\begin{thebibliography}{10}

\bibitem{MIMO}
E.~G. Larsson, O.~Edfors, F.~Tufvesson, and T.~L. Marzetta, ``Massive {MIMO}
  for next generation wireless systems,'' {\em IEEE Commun. Mag.}, vol.~52,
  pp.~186--195, Feb. 2014.

\bibitem{POLAR}
D.~Hui, S.~Sandberg, Y.~Blankenship, M.~Andersson, and L.~Grosjean, ``Channel
  coding in {5G} new radio: A tutorial overview and performance comparison with
  {4G} {LTE},'' {\em IEEE Veh. Technol. Mag.}, vol.~13, pp.~60--69, Dec 2018.

\bibitem{NOMA}
J.~Dai, K.~Niu, Z.~Si, C.~Dong, and J.~Lin, ``Polar-coded non-orthogonal
  multiple access,'' {\em IEEE Trans. Signal Processing}, vol.~66,
  pp.~1374--1389, March 2018.

\bibitem{Bandshort}
T.~S. Rappaport, S.~Sun, R.~Mayzus, H.~Zhao, Y.~Azar, K.~Wang, G.~N. Wong,
  J.~K. Schulz, M.~Samimi, and F.~Gutierrez, ``Millimeter wave mobile
  communications for {5G} cellular: It will work!,'' {\em IEEE Access}, vol.~1,
  pp.~335--349, 2013.

\bibitem{3GPPFR1}
{3GPP TS 38.101-1}, ``User equipment radio transmission and reception; part 1:
  range 1 standalone (release 15),'' Dec. 2017.

\bibitem{Pathloss}
I.~A. Hemadeh, K.~Satyanarayana, M.~El-Hajjar, and L.~Hanzo, ``Millimeter-wave
  communications: Physical channel models, design considerations, antenna
  constructions, and link-budget,'' {\em IEEE Commun. Surveys Tuts.}, vol.~20,
  pp.~870--913, Secondquarter 2018.

\bibitem{Necessary}
Y.~Kim, H.~Lee, P.~Hwang, R.~K. Patro, J.~Lee, W.~Roh, and K.~Cheun,
  ``Feasibility of mobile cellular communications at millimeter wave
  frequency,'' {\em IEEE J. Sel. Topics Signal Processing}, vol.~10,
  pp.~589--599, April 2016.

\bibitem{FD1}
A.~Wiesel, Y.~C. Eldar, and S.~Shamai, ``Zero-forcing precoding and generalized
  inverses,'' {\em IEEE Trans. Signal Processing}, vol.~56, pp.~4409--4418,
  Sep. 2008.

\bibitem{FD2}
A.~B. Gershman, N.~D. Sidiropoulos, S.~Shahbazpanahi, M.~Bengtsson, and
  B.~Ottersten, ``Convex optimization-based beamforming,'' {\em IEEE Signal
  Processing Mag.}, vol.~27, pp.~62--75, May 2010.

\bibitem{Sparse}
O.~E. Ayach, S.~Rajagopal, S.~Abu-Surra, Z.~Pi, and R.~W. Heath, ``Spatially
  sparse precoding in millimeter wave {MIMO} systems,'' {\em IEEE Trans.
  Wireless Commun.}, vol.~13, pp.~1499--1513, Mar. 2014.

\bibitem{MF}
L.~Lu, G.~Y. Li, A.~L. Swindlehurst, A.~Ashikhmin, and R.~Zhang, ``An overview
  of massive {MIMO}: Benefits and challenges,'' {\em IEEE J. Sel. Topics Signal
  Processing}, vol.~8, pp.~742--758, Oct 2014.

\bibitem{ZFR}
C.~B. Peel, B.~M. Hochwald, and A.~L. Swindlehurst, ``A vector-perturbation
  technique for near-capacity multiantenna multiuser communication-part i:
  channel inversion and regularization,'' {\em IEEE Trans. Commun.}, vol.~53,
  pp.~195--202, Jan 2005.

\bibitem{AP1}
C.~H. Doan, S.~Emami, D.~A. Sobel, A.~M. Niknejad, and R.~W. Brodersen,
  ``Design considerations for 60 {GHz} {CMOS} radios,'' {\em IEEE Commun.
  Mag.}, vol.~42, pp.~132--140, Dec 2004.

\bibitem{AP2}
W.~Roh, J.~Seol, J.~Park, B.~Lee, J.~Lee, Y.~Kim, J.~Cho, K.~Cheun, and
  F.~Aryanfar, ``Millimeter-wave beamforming as an enabling technology for {5G}
  cellular communications: theoretical feasibility and prototype results,''
  {\em IEEE Commun. Mag.}, vol.~52, pp.~106--113, February 2014.

\bibitem{AP3}
S.~Hur, T.~Kim, D.~J. Love, J.~V. Krogmeier, T.~A. Thomas, and A.~Ghosh,
  ``Millimeter wave beamforming for wireless backhaul and access in small cell
  networks,'' {\em IEEE Trans. Commun.}, vol.~61, pp.~4391--4403, October 2013.

\bibitem{stad}
T.~Nitsche, C.~Cordeiro, A.~B. Flores, E.~W. Knightly, E.~Perahia, and J.~C.
  Widmer, ``{IEEE} 802.11ad: directional 60 {GHz} communication for
  multi-{Gigabit}-per-second {Wi-Fi} [invited paper],'' {\em IEEE Commun.
  Mag.}, vol.~52, pp.~132--141, December 2014.

\bibitem{BADAP}
A.~Alkhateeb, G.~Leus, and R.~W. Heath, ``Limited feedback hybrid precoding for
  multi-user millimeter wave systems,'' {\em IEEE Trans. Wireless Commun.},
  vol.~14, pp.~6481--6494, Nov 2015.

\bibitem{HPSERVY}
A.~F. Molisch, V.~V. Ratnam, S.~Han, Z.~Li, S.~L.~H. Nguyen, L.~Li, and
  K.~Haneda, ``Hybrid beamforming for massive {MIMO}: A survey,'' {\em IEEE
  Commun. Mag.}, vol.~55, pp.~134--141, Sep. 2017.

\bibitem{ARVcs}
Y.~Lee, C.~Wang, and Y.~Huang, ``A hybrid {RF}/baseband precoding processor
  based on parallel-index-selection matrix-inversion-bypass simultaneous
  orthogonal matching pursuit for millimeter wave {MIMO} systems,'' {\em IEEE
  Trans. Signal Processing}, vol.~63, pp.~305--317, Jan 2015.

\bibitem{DFT1}
M.~Kim and Y.~H. Lee, ``Mse-based hybrid {RF}/baseband processing for
  millimeter-wave communication systems in {MIMO} interference channels,'' {\em
  IEEE Trans. Veh. Technol.}, vol.~64, pp.~2714--2720, June 2015.

\bibitem{DFT2}
J.~Brady, N.~Behdad, and A.~M. Sayeed, ``Beamspace {MIMO} for millimeter-wave
  communications: System architecture, modeling, analysis, and measurements,''
  {\em IEEE Trans. Antennas and Propag.}, vol.~61, pp.~3814--3827, July 2013.

\bibitem{optfhp}
F.~Sohrabi and W.~Yu, ``Hybrid digital and analog beamforming design for
  large-scale antenna arrays,'' {\em IEEE J. Sel. Topics Signal Processing},
  vol.~10, pp.~501--513, April 2016.

\bibitem{fuldecp1}
L.~Liang, W.~Xu, and X.~Dong, ``Low-complexity hybrid precoding in massive
  multiuser {MIMO} systems,'' {\em IEEE Wireless Commun. Lett.}, vol.~3,
  pp.~653--656, Dec 2014.

\bibitem{fuldecp2}
D.~H.~N. Nguyen, L.~B. Le, T.~Le-Ngoc, and R.~W. Heath, ``Hybrid {MMSE}
  precoding and combining designs for {mmWave} multiuser systems,'' {\em IEEE
  Access}, vol.~5, pp.~19167--19181, 2017.

\bibitem{fuladd}
Y.~Ahn, T.~Kim, and C.~Lee, ``A beam steering based hybrid precoding for
  {MU-MIMO} {mmWave} systems,'' {\em IEEE Commun. Lett.}, vol.~21,
  pp.~2726--2729, Dec. 2017.

\bibitem{fuladd2}
D.~H.~N. Nguyen, L.~B. Le, and T.~Le-Ngoc, ``Hybrid {MMSE} precoding for
  {mmWave} multiuser {MIMO} systems,'' in {\em 2016 IEEE Int. Commun. Conf.
  (ICC)}, pp.~1--6, May 2016.

\bibitem{subpso}
O.~Alluhaibi, Q.~Z. Ahmed, J.~Wang, and H.~Zhu, ``Hybrid digital-to-analog
  precoding design for mm-wave systems,'' in {\em 2017 IEEE Int. Commun. Conf.
  (ICC)}, pp.~1--6, May 2017.

\bibitem{subsnr}
N.~Li, Z.~Wei, H.~Yang, X.~Zhang, and D.~Yang, ``Hybrid precoding for {mmWave}
  massive {MIMO} systems with partially connected structure,'' {\em IEEE
  Access}, vol.~5, pp.~15142--15151, 2017.

\bibitem{subcodbk1}
C.~Kim, T.~Kim, and J.~Seol, ``Multi-beam transmission diversity with hybrid
  beamforming for {MIMO-OFDM} systems,'' in {\em 2013 IEEE Globecom Workshops
  (GC Wkshps)}, pp.~61--65, Dec 2013.

\bibitem{subcodbk2}
J.~Singh and S.~Ramakrishna, ``On the feasibility of codebook-based beamforming
  in millimeter wave systems with multiple antenna arrays,'' {\em IEEE Trans.
  Wireless Commun.}, vol.~14, pp.~2670--2683, May 2015.

\bibitem{subcodbk3}
A.~Li and C.~Masouros, ``Hybrid analog-digital millimeter-wave {MU-MIMO}
  transmission with virtual path selection,'' {\em IEEE Commun. Lett.},
  vol.~21, pp.~438--441, Feb. 2017.

\bibitem{subdecp}
A.~Li and C.~Masouros, ``Hybrid precoding and combining design for
  millimeter-wave multi-user {MIMO} based on {SVD},'' in {\em 2017 IEEE Int.
  Commun. Conf. (ICC)}, pp.~1--6, May 2017.

\bibitem{subml1}
Y.~Long, Z.~Chen, J.~Fang, and C.~Tellambura, ``Data-driven-based analog beam
  selection for hybrid beamforming under mm-wave channels,'' {\em IEEE J. Sel.
  Topics Signal Processing}, vol.~12, pp.~340--352, May 2018.

\bibitem{subml2}
X.~Gao, L.~Dai, Y.~Sun, S.~Han, and I.~Chih-Lin, ``Machine learning inspired
  energy-efficient hybrid precoding for {mmWave} massive {MIMO} systems,'' in
  {\em 2017 IEEE Int. Commun. Conf. (ICC)}, pp.~1--6, May 2017.

\bibitem{AHPsu}
S.~Park, A.~Alkhateeb, and R.~W. Heath, ``Dynamic subarrays for hybrid
  precoding in wideband {mmWave} {MIMO} systems,'' {\em IEEE Trans. Wireless
  Commun.}, vol.~16, pp.~2907--2920, May 2017.

\bibitem{AHPmu}
X.~Zhu, Z.~Wang, L.~Dai, and Q.~Wang, ``Adaptive hybrid precoding for multiuser
  massive {MIMO},'' {\em IEEE Commun. Lett.}, vol.~20, pp.~776--779, Apr. 2016.

\bibitem{AHPmu2}
Q.~Yu, X.~Zhai, and M.~Zhao, ``An energy-efficient hybrid precoding algorithm
  for multiuser {mmWave} massive {MIMO} systems,'' in {\em 2017 IEEE 86th Veh.
  Technol. Conf. (VTC-Fall)}, pp.~1--5, Sep. 2017.

\bibitem{MYpimrc}
L.~Zhang, L.~Gui, Q.~Qin, and Y.~Tu, ``Adaptively-connected structure for
  hybrid precoding in multi-user massive {MIMO} systems,'' in {\em 2018 IEEE
  29th Annual Int. Symp. Pers. Indoor Mobile Radio Commun. (PIMRC)}, pp.~1--7,
  Sep. 2018.

\bibitem{HAC}
S.~Zhou, Z.~Xu, and F.~Liu, ``Method for determining the optimal number of
  clusters based on agglomerative hierarchical clustering,'' {\em IEEE Trans.
  Neural Netw. Learn. Syst.}, vol.~28, pp.~3007--3017, Dec 2017.

\bibitem{AMmethod}
X.~Yu, J.~Shen, J.~Zhang, and K.~B. Letaief, ``Alternating minimization
  algorithms for hybrid precoding in millimeter wave {MIMO} systems,'' {\em
  IEEE J. Sel. Topics Signal Processing}, vol.~10, pp.~485--500, April 2016.

\bibitem{SVD}
M.~Shehata, M.~Crussi¨¨re, M.~H¨¦lard, and P.~Pajusco, ``Leakage based users
  selection for hybrid beamforming in {MillimeterWave} {MIMO},'' in {\em 2018
  IEEE 29th Annual Int. Symp. Pers. Indoor Mobile Radio Commun. (PIMRC)},
  pp.~1144--1150, Sep. 2018.

\bibitem{KM}
A.~K. Jain, ``Data clustering: 50 years beyond k-means,'' in {\em Machine
  Learning and Knowledge Discovery in Databases} (W.~Daelemans, B.~Goethals,
  and K.~Morik, eds.), (Berlin, Heidelberg), pp.~3--4, Springer Berlin
  Heidelberg, 2008.

\bibitem{OPP}
J.~C. Gower, G.~B. Dijksterhuis, {\em et~al.}, {\em Procrustes problems},
  vol.~30.
\newblock Oxford University Press on Demand, 2004.

\bibitem{Power}
R.~M¨¦ndez-Rial, C.~Rusu, A.~Alkhateeb, N.~Gonz¨¢lez-Prelcic, and R.~W. Heath,
  ``Channel estimation and hybrid combining for {mmWave}: Phase shifters or
  switches?,'' in {\em 2015 Inf. Theory Appl. Workshop (ITA)}, pp.~90--97, Feb
  2015.

\end{thebibliography}
}









\end{document}